\documentclass[10pt,journal,compsoc]{IEEEtran}

\ifCLASSOPTIONcompsoc
\usepackage[nocompress]{cite}
\else
\usepackage{cite}
\fi

\usepackage{algorithm,amsbsy,amsmath,amssymb,epsfig,bbm,mathrsfs,multirow,amsthm,xcolor,subcaption}
\usepackage[spaces,hyphens]{url}

\usepackage{hyperref}
\hypersetup{
	colorlinks=true,
	linkcolor=black,
	filecolor=black,   
	urlcolor=black, citecolor=black,}

\usepackage{algorithmic}
\newtheorem{theorem}{Theorem}

\DeclareMathAlphabet\mathbfcal{OMS}{cmsy}{b}{n}
\usepackage{array}

\makeatletter
\def\endthebibliography{%
	\def\@noitemerr{\@latex@warning{Empty `thebibliography' environment}}%
	\endlist
}
\makeatother

\usepackage{setspace}

\usepackage[hyphens]{url}
\usepackage{hyperref}
\usepackage[hyphenbreaks]{breakurl}

\begin{document}
	%
	\title{MetaSlicing: A Novel Resource Allocation Framework for Metaverse}
	%
	%
	%
	
	\author{Nam~H.~Chu,
		Dinh~Thai~Hoang,
		Diep~N.~Nguyen,	
		Khoa~T.~Phan, \\		
		Eryk~Dutkiewicz,
		Dusit Niyato,
		and Tao~Shu
		\thanks{Nam~H.~Chu, Dinh~Thai~Hoang, Diep~N.~Nguyen, and Eryk~Dutkiewicz are with the School of Electrical and Data Engineering, University of Technology Sydney, Australia (e-mails: namhoai.chu@student.uts.edu.au, hoang.dinh@uts.edu.au, diep.nguyen@uts.edu.au, and eryk.dutkiewicz@uts.edu.au).}
		\thanks{Khoa~T.~Phan is with School of Engineering and Mathematical Sciences, Department of Computer Science and Information Technology, La Trobe University, Melbourne, Australia (e-mail: K.Phan@latrobe.edu.au).}
		\thanks{Dusit Niyato is with the School of Computer Science and Engineering, Nanyang Technological University, Singapore 639798 (e-mail: dniyato@ntu.edu.sg).}
		\thanks{Tao Shu is with the Department of Computer Science and Software Engineering, Auburn University, Auburn, AL 36849. (e-mail: tshug@auburn.edu).}
	}%
		

\IEEEtitleabstractindextext{%
	\begin{abstract}
  		Creating and maintaining the Metaverse requires enormous resources that have never been seen before, especially computing resources for intensive data processing to support the Extended Reality, enormous storage resources, and massive networking resources for maintaining ultra high-speed and low-latency connections.
  		Therefore, this work aims to propose a novel framework, namely MetaSlicing, that can provide a highly effective and comprehensive solution in managing and allocating different types of resources for Metaverse applications. 
  		In particular, by observing that Metaverse applications may have common functions, we first propose grouping applications into clusters, called MetaInstances. In a MetaInstance, common functions can be shared among applications.
  		As such, the same resources can be used by multiple applications simultaneously, thereby enhancing resource utilization dramatically.
  		To address the real-time characteristic and resource demand's dynamic and uncertainty in the Metaverse, we develop an effective framework based on the semi-Markov decision process and propose an intelligent admission control algorithm that can maximize resource utilization and enhance the Quality-of-Service for end-users.
  		Extensive simulation results show that our proposed solution outperforms the Greedy-based policies by up to 80\% and 47\% in terms of long-term revenue for Metaverse providers and request acceptance probability, respectively.
  		  		
	\end{abstract}
	
}
	
	%
	\maketitle
	\IEEEdisplaynontitleabstractindextext
	\IEEEpeerreviewmaketitle

\IEEEraisesectionheading{\section{Introduction}}
	\label{sec:intro}
	\IEEEPARstart{A}lthough the Metaverse's concept first appeared in 1992~\cite{stephenson_1992}, it has just been attracting more attention from academia and industry in the last few years, thanks to the recent advances in technologies (e.g., extended reality, 5G/6G networks, and edge intelligence) along with great efforts of many big corporations such as Facebook~\cite{facebook_2021} and Microsoft~\cite{microsoft_2021}.
		The Metaverse is expected to bring a new revolution to the digital world.
		Unlike existing virtual worlds (e.g., Second Life and Roblox), where the users' presentations (e.g., avatars/characters) and assets are limited in specific worlds, the Metaverse can be realized as a seamless integration of multiple virtual worlds~\cite{youtube_metaverse}.
		Each virtual world in the Metaverse can be created for a certain application, such as entertainment, education, and healthcare.
		Similar to our real lives, Metaverse users can bring their assets from one to another virtual world while preserving their values, and vice versa.
		Moreover, the Metaverse is expected to further integrate digital and physical worlds, e.g., digitizing the physical environment by the digital twin~\cite{wu_digital_2021}.		
		For example, in the Metaverse, we can create our virtual objects, such as outfits and paintings, and then bring them to any virtual world to share or trade with others.
		We also can share virtual copies of a real object in different virtual worlds.	
		Thus, the Metaverse will bring total new experiences that can change many aspects of our daily lives, such as entertainment, education, e-commerce, healthcare, and smart industries~\cite{yu_healthcare_2012, kwon_smart_2021, jeong_ecommerce_2022}.	
	
	However, extremely-high resource demand in the Metaverse is one of the biggest challenges that is impeding the deployment of the Metaverse~\cite{xu2022full}.
		To fulfil the Quality-of-Service (QoS) and user experience requirements in the Metaverse, it demands enormous resources that may have never been seen before. 
		First, the Metaverse is expected to support millions of users simultaneously since each Metaverse application can host a hundred thousand users simultaneously.
		For example, the peak number of concurrent players of Counter Strike - Global Offensive is more than one million in 2021~\cite{clement_steam_2022}.
		It is forecasted that data usage on networks can be expanded more than $20$ times by the operation of Metaverse~\cite{credit_metaverse_2022}.
		Second, the Extended Reality (XR) technology is believed to be integrated into Metaverse's applications such that users can interact with virtual and physical objects via their digital avatars, e.g., digital twin~\cite{xu2022full}. 
		Therefore, the Metaverse requires extensive computing to render three dimensional (3-D) objects, a large amount of data collected from perceived networks, e.g., Internet of Things (IoT), and an ultra-low delay communication to maintain a seamless user experience.
		Third, unlike the current online platforms (e.g., massive multiplayer online role-playing games where the uplink throughput can be much lower than that of downlink throughput~\cite{wang_characterizing_2012}), the Metaverse requires extremely-high throughput for both uplink and downlink transmission links.
		The reason is that Metaverse users can create their digital objects and then share/trade them via this innovation platform. 		
		Therefore, the Metaverse's demand for resources (e.g., computing, networking, and storage) likely exceeds that of any existing massive multiplayer online application~\cite{xu2022full}.
	
	In this context, although deploying the Metaverse on the cloud is a possible solution, it leads to several challenges.
		First, the cloud is often located in a physical area (e.g., a data center), making it potentially a point-of-congestion when millions of users connect at once.
		Second, since users come from around the world, a huge amount of exchanged data puts stress on the communication infrastructure.
		This results in high delay, which severely impacts the Metaverse since the delay is one of the crucial drivers of user experience~\cite{zhao_estimating_2017}.	
		In this context, multi-tier resource allocation architecture, where the computing, storage, networking, and communication capabilities are distributed along the path from end-users to the cloud, is a promising solution for the Metaverse implementation.		
	
	In the literature, there are only a few attempts to investigate the Metaverse resource management~\cite{jiang2021reliable, xu2021wireless, ng2021unified, han2021dynamic}.
		Specifically, in~\cite{jiang2021reliable}, the authors consider computing resource allocation for a single-edge computing architecture that has a limited computing resource to allocate for some nearby Metaverse users.
		Similarly, in~\cite{xu2021wireless} and \cite{ng2021unified}, resource allocation at the edge is considered, but more resource types, i.e., computation and communication, are considered.
		In particular, the work in~\cite{xu2021wireless} proposes a pricing model-based resource management to accelerate the trading of Virtual Reality (VR) services between end-users and VR service providers.
		In~\cite{ng2021unified}, the authors address the stochastic demand problem for an application of education in the Metaverse. 
		Specifically, they propose a stochastic optimal resource allocation method to minimize the cost for the virtual service provider.	
		Unlike the above works, in~\cite{han2021dynamic}, the authors propose an evolutionary game-based resource management for perception networks (e.g., IoT) that are used to collect data for the Metaverse. 
				
	It can be observed that none of the above studies considers the multi-tier computing architecture for resource allocation problem.
		Instead, their approaches are only appropriate for a single-tier edge computing architecture~\cite{xu2021wireless, ng2021unified, jiang2021reliable, han2021dynamic}.
		However, as analyzed above, due to the extremely high resource demands of Metaverse applications, the single-tier computing resource model may not be appropriate and effective.		
		Furthermore, it can be observed that in the Metaverse, many Metaverse applications may share some common functions.		
		For example, a digital map is indeed a common function between tourism and navigation applications.
		Currently, sharing functions among applications has already been made. 
		For instance, Google Map's Application Programming Interface (API) provides various functions (e.g., digital map, check-in, display live data synching with location~\cite{google_map_api}) that can be shared among many applications, e.g., Pokemon Go~\cite{pokemon}, Wooorld~\cite{wooorld}, and CNN iReport Map\cite{cnn_map}.
		This special feature of the Metaverse indeed can be leveraged to maximize resource utilization for Metaverse applications. 
		However, none of the above works can exploit the similarity among applications to improve resource utilization for the Metaverse.
		Moreover, in practice, users join and leave the Metaverse at any time, leading to the high uncertainty and dynamic of resource demands.
		Among the aforementioned works, only the study in~\cite{ng2021unified} addresses the stochastic demands of users.
		Nevertheless,  it only considers resource allocation for a Metaverse education application. 
		In addition, this approach is only appropriate for a single-tier resource allocation and could not leverage the similarities among Metaverse applications' functions in order to maximize system performance. 
		Thus, there is an urgent need for an effective and comprehensive solution for the Metaverse to handle not only the massive resource usage but also the dynamic and uncertain resource demand.	
	
	To address all the aforementioned challenges, we propose a novel framework, namely MetaSlicing, to intelligently allocate diverse types of resources for Metaverse applications by analyzing the incoming requests and allocating appropriate resources, and thereby maximizing the whole system performance.
		Firstly, we introduce the idea of decomposing an application into multiple functions to facilitate the deployment and management of Metaverse, which are highly complex.		 
		In particular, each function of an application can be initialized separately and placed at a different tier in the system according to functions' requirements and tiers' available resources.
		For example, functions with low latency requirements can be placed at a low tier (e.g., tier-1), while those with low update frequency can be placed at a higher tier.
		By doing so, the application decomposition can not only provide a flexible solution for deploying Metaverse applications but also  utilize all networks' resources from different tiers.
		Secondly, we propose a novel technique, called MetaInstance, to address extreme-high resource demands of the Metaverse.
		The MetaInstance aims to improve resource utilization by exploiting the similarities among Metaverse applications. 
		To be more specific, applications with common functions will be grouped into a MetaInstance, and the common functions will be shared among these applications instead of creating one for each application. 
		Therefore, this technique can save more resources.
		Finally, to address the uncertainty and dynamic of resource demands as well as the real-time response of Metaverse applications, we propose a highly-effective semi-Markov decision process-based framework together with a reinforcement learning algorithm to automatically learn and find the optimal policy for the system under dynamic and uncertain resource demand.
		
	The main contributions of this work are summarized as follows:
 	\begin{itemize}
 		\item We propose a novel framework in which different types of resources at different tiers of the computing architecture can be allocated smartly to maximize the system performance for the Metaverse.
 		\item We introduce two innovative techniques, including Metaverse application decomposition and MetaInstance, to maximize resource utilization for the proposed multi-tier computing-based Metaverse.  		
 		\item We propose a highly effective admission control model based on the semi-Markov decision process that can capture the high dynamic and uncertainty of resource demand as well as the real-time characteristic of the Metaverse.
 		\item We develop an intelligent algorithm that can automatically find the optimal admission control policy for the Metaverse without requiring the complete information about the dynamic and uncertainty of resource demand.
 		\item We perform extensive simulations not only to explore the resilience of our proposed framework but also to gain insights into the key factors that can affect the system performance.  
 	\end{itemize}
 
%

	\section{System Model}
	\label{sec:model}
	\begin{figure*}[t]
		\centering
		\includegraphics[width=0.95\linewidth]{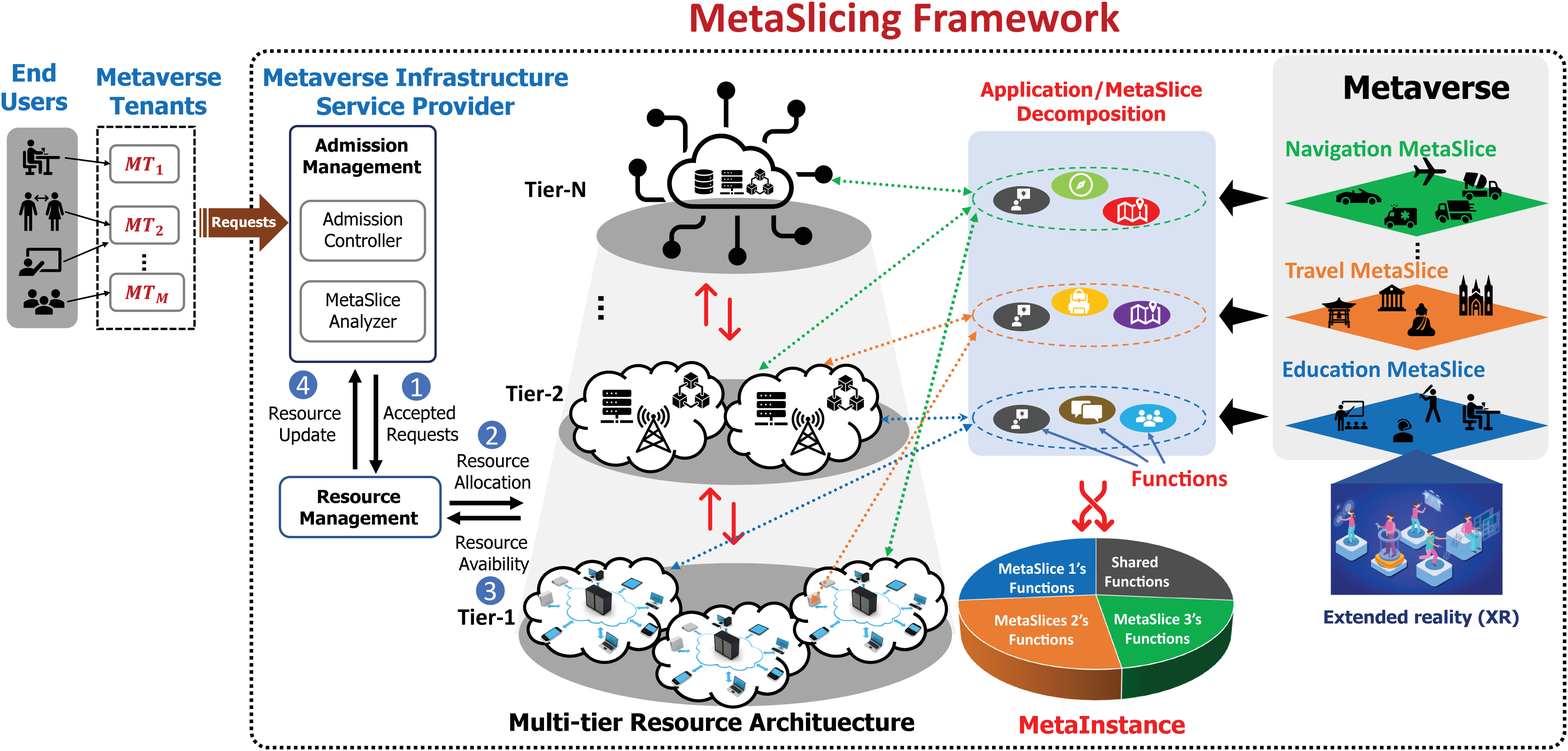}
		\caption{The system model of the proposed MetaSlicing framework. In this framework, different resource types in different tiers can be used and shared to create Metaverse applications (i.e., MetaSlices)}
		\label{fig.system_model}
	\end{figure*}
	In this paper, we consider a system model including three main parties, i.e., (i) End-users, (ii) Metaverse tenants, and (iii) the Metaverse Infrastructure Service Provider (MISP), as illustrated in Fig.~\ref{fig.system_model}.
		First, an end-user subscribes to a Metaverse tenant to request a Metaverse application, namely MetaSlice.
		Then, the Metaverse tenant will request the MetaSlice from the MISP according to its subscribed users' demands.
		If a request is accepted, the MISP will allocate its resources to initiate this MetaSlice.	
		In the following, we explain the main components together with their interactions in our proposed framework in more details.
				
	\subsection{MetaSlicing: Dynamic Resource Allocation Framework for Metaverse}	
	\label{subsec:metaslicing}
	As discussed in the previous section, Metaverse applications usually demand extremely high resources, and thus we propose a novel multi-tier resource allocation architecture that can effectively and dynamically allocate resources for MetaSlices. 	
		First, this architecture alleviates the extensive resource demands of MetaSlices for both the Metaverse tenants and end-users.
		Second, this architecture enables the distribution of different types of resources, e.g., computing, storage, networking, and communication capabilities, along the path from end-users to the cloud.
		By doing so, we can leverage the resources placed near end-users, resulting in a low delay and high QoS for users.
		Moreover, distributing resources increases the resilience of the system compared to the traditional centralized cloud-based resource allocation architecture.

		To maximize resource usage efficiency, we first observe that a MetaSlice involves multiple functions, which can operate independently. 
		For example, in a travel MetaSlice, a recommendation function can suggest attractive places to users based on their locations or preferences.
		Then, users can use a path-finding function to find the best route to these places or create a tour of several places.
		Thus, these functions are independent and can be designed to work separately to make more conveniences for maintaining and upgrading purposes.
		As such, this paper considers that a MetaSlice can be decomposed into multiple independent functions. 
		Note that this work does not focus on optimal application decomposition, and we assume that applications can be decomposed by using existing methods, e.g.,~\cite{Sahhaf2015Network, alturki2019Exploring}. 
		We consider that each function is allocated dedicated resources due to the strict QoS requirements of Metaverse applications.   
		In particular, a function with the dedicated resource allocation scheme likely executes faster than this function under the dynamic scheme. 
		This is because, under the dynamic resource allocation scheme, a resource allocation procedure is performed whenever a request for processing arrives at a function, leading to an unavoidable additional delay that may degrade the QoS of Metaverse users.
		In addition, because resources are not reserved for each function in a dynamic allocation scheme, a function may not have sufficient  resources to support an application's request that arrives when a server is overloading with other functions, leading to a high delay or even service disruption.
		This problem can be alleviated by migrating the function to another server with adequate resources; however, it introduces another additional delay to the function execution.				
		In this context, dedicated resource allocation can avoid this issue by reserving resources for each function.
		Moreover, the study in~\cite{wolke2014planning} points out that under a typical load of business applications, the dedicated resource allocation can achieve a better energy efficiency than that of dynamic resource allocation.
	
	In practice, the function decomposition technique provides flexibility for Metaverse implementation.
		For example, a MetaSlice for travel may consist of several major functions such as a digital map, real-time traffic, real-time weather, and a driving assistant.
		In this case, these functions can be placed dynamically in the multi-tier resource allocation architecture depending on the functions' requirements.
		For instance, real-time traffic and driving assistant functions can be placed at tier-1 near end-users since they require low delay, while a digital map (that does not need frequent updates) can be located at a high tier, e.g., at the cloud.		
		Thus, the functional decomposition technique offers a flexible and effective implementation of multi-tier based MetaSlices. 
		To support the decomposition technique, we consider that a MetaSlice is created based on the MetaBlueprint, i.e., a template describing the workflow, configuration, and structure for initializing and managing this MetaSlice during its life cycle. 
		From a technical standpoint, a MetaSlice (i.e., an application in the Metaverse) can be analogous to the network slice paradigm in the fifth generation of the cellular network (5G) that consists of multiple network functions~\cite{ngmn_nsi_2016}.
		  	 
	We consider different types of MetaSlices, e.g., tourism, education, industry, and navigation.
		MetaSlices can be grouped into $I$ classes based on their characteristics, such as occupied resources, technical configurations, and QoS.		
		For example, navigation MetaSlice may have driving assistant functions requiring ultra-low latency and highly-reliable connections, while security and resilience are among the top concerns of e-commerce and industry MetaSlices.			
		In addition, different types of MetaSlice may share the same functions.
		For example, tourism and navigation MetaSlices can use the same underlying digital map and the real-time traffic and weather functions whose data are collected by the same perception network, e.g., IoT.
			
	Moreover, we can observe that different Metaverse tenants can create/manage multiple variants from the same type of MetaSlice (e.g., education, industry, or navigation). 
		Therefore, it is likely that ongoing MetaSlices may share the same functions.    
		In this case, a lot of resources can be shared, leading to greater resource utilization and higher revenue for the MISP, {similar to that in Network Slicing~\cite{ngmn_nsi_2016}}.	
		Based on this fact, we consider that Metaverse applications can be classified into groups, namely MetaInstances.
		Thus, a MetaInstance can be defined by two function types, i.e., (i) shared function and (ii) dedicated function belonging to specific MetaSlices, as illustrated in Fig.~\ref{fig.system_model}.				
		In this case, a MetaInstance can maintain a function configuration consisting of a list of functions and a description of interactions among them.
		Because the capability of a function is limited, sharing a function for too many MetaSlices may lead to a decrease in user experience (e.g., processing delay) or even service disruption. 
		Therefore, in practice, a function can be only shared by a maximum number of $N_L$ MetaSlices. 
		From the technical perspective, the implementation of a MetaInstance can be similar to that of the Network Slice Instance (NSI), where network slices can share some Network Functions (NFs)~\cite{ngmn_nsi_2016}.
 
		It is worth mentioning that even though the MetaSlicing may be similar to the network slicing~\cite{ngmn_nsi_2016}, they are actually not the same.
		In particular, network slicing aims to address the diversity (or even conflict) in communication requirements among various businesses by running multiple logical networks (i.e., Network Slices) over a physical network.
		For example, one (e.g., automotive customers) may require ultra-low latency connections while others (e.g., manufacturing customers) require ultra-reliable connections. 
		Thus, the network slicing focuses on providing diverse types of communications. 
		In contrast, MetaSlicing is a framework to effectively manage resources in a multi-tier computing architecture by decomposing an application into functions and then optimally distributing them at different tiers. 
		In addition, as explained in the previous paragraphs, due to different features of Metaverse applications, resource allocation scheme for MetaSlicing is also designed different from that of the network slicing, where dynamic resource allocation approaches are preferred~\cite{zhang2017network}. 
			
	Based on the aforementioned analysis, we can observe that, on the one hand, our proposed MetaSlicing framework can offer a great solution to the MISP by maximizing the resource utilization and at the same time minimizing the deployment cost and initialization time for Metaverse's applications. 
		On the other hand, this framework can also benefit end-users by achieving greater user experience, e.g., lower delay and more reliable services.
		To achieve these results, the Admission Controller in MetaSlicing plays a critical role.
		For example, accepting requests of MetaSlices that share some functions with the ongoing MetaSlices may help the system to save more resources than accepting those with less or without sharing functions with the MetaSlices running in the system.
		In addition, Resource Management is another important factor determining the Metaverse system performance.
		In the following subsection, we explain these components in our proposed MetaSlicing framework.	
	
	\subsection{Admission Control and Resource Management}
	\label{subsec:addmission_control}	
	Recall that according to the demands of subscribed end-users, a Metaverse tenant sends a MetaSlice request associated with MetaBlueprint to the MISP. Then, the Admission Management and Resource Management are executed as follows.		
		As shown in Fig.~\ref{fig.system_model}, the Admission Management block of the MISP consists of a MetaSlice Analyzer and an Admission Controller.
		Once receiving a MetaSlice request, the MetaSlice Analyzer will analyze the request's MetaBlueprint to determine similarities between the functions and configuration of the requested MetaSlice and those of the ongoing MetaInstances.
		In this paper, we consider that a MetaBlueprint consists of at least (i) the function configuration record, including the list of required functions and the description of interactions among them, and (ii) the MetaSlice configuration (e.g., class ID and required resources).	
		Based on the similarity analysis (to be presented in Section IV) obtained from the MetaSlice Analyzer and the currently available resources of the system, the Admission Controller decides whether to accept the request or not according to its admission policy.
	
	Suppose that a MetaSlice request is accepted, the Resource Management allocates the accepted MetaSlice to a MetaInstance with the highest similarity index and updates this MetaInstance accordingly.
		Specifically, if the accepted MetaSlice has dedicated functions, system resources are allocated to initiate these new functions.
		If the current MetaInstances do not share any function with the new MetaSlice, a new MetaInstance is created for the accepted MetaSlice.
		When a MetaSlice departs/completes, its resources will be released, and the MetaInstance will be updated accordingly.		 
		
	In this paper, we consider $D$ resource types owned by the MISP, e.g., computing, networking, and storage.
		Then, the required resources for a MetaSlice $m$ can be represented by a resource vector, i.e., \mbox{$\mathbf{n}_m=[n_m^1, \dots, n_m^d,\dots, n_m^D]$}, where $n_m^d$ is the amount of type $d$ resources.
			In this case, the total occupied resources by all MetaSlices cannot exceed the maximum resources of MISP, i.e., 
		\begin{equation}
			\sum_{m \in \mathcal{M}}n_m^d \leq N^d, \quad \forall d \in \{1,\dots, D\},
		\end{equation}
		where $N^d$ is the total amount of type $d$ resources of the MISP, and $\mathcal{M}$ is the set of all running MetaSlices in the system.
		In our proposed solution, the system's available resources and the required resources for the request are two crucial factors for the admission control in MetaSlicing. 
		However, in practice, the future requests' arrival process and its required resources are likely unknown in advance.
		In addition, the departure process of MetaSlices (i.e., how long a MetaSlice remains in the system) is also highly dynamic and uncertain.  
		Therefore, in the next section, we will introduce a framework based on the semi-Markov decision process to address these challenges.
	
	\section{MetaSlicing Admission Control Formulation}
	\label{sec:formulation}
	In this paper, we propose a highly-effective semi-Markov Decision Process (sMDP) framework to address the MetaSlice admission control problem due to the following reasons. 
		First, the sMDP can enable the MetaSlicing's Admission Controller to adaptively make the best decisions (i.e., whether to accept or reject a MetaSlice request) based on the current available system resources (i.e., computing, networking, and storage) and the MetaSlice request's blueprint (i.e., resource, class and similarity) without requiring complete information about the surrounding environment (e.g., arrival and departure processes of MetaSlices) to maximize the MISP's long-term revenue.		
		Second, in practice, MetaSlice requests can arrive at any time, so the admission decision needs to be made as soon as possible.
		However, the conventional Markov Decision Process (MDP) only takes an action at each time slot with an equal time period, making it unable to capture real-time events, e.g., request arrival~\cite{tijms_a_2003}.
		In contrast, the sMDP makes a decision whenever an event occurs so that it can perfectly capture the real time of MetaSlicing.
		Finally, the MetaSlice's lifetime is highly uncertain.
		Upon a MetaSlice departs from the system, its occupied resources are released, and the system state transits to a new state immediately.
		Again, the conventional MDP is unable to capture this transition as it works in a discrete-time fashion.			
		
	To that end, the sMDP will be used in our framework to enable the MISP to make real-time decisions and maximize its long-term revenues. 
		Technically, an sMDP can be defined by a set of five components, including (i) the decision epoch $t_i$, (ii) the state space $\mathcal{S}$, (iii) the action space $\mathcal{A}$, (iv) the transition probability $\mathcal{T}$, and (v) the immediate reward function $r$. 
		In the following sections, we will explain how this framework can capture all events in the MetaSlice system and make optimal decisions for the MISP.
			
		\subsection{Decision Epoch}
		The decision epochs are defined as points of time at which decisions are made~\cite{tijms_a_2003}.
		In our real-time MetaSlicing system, the Admission Controller must make a decision once a MetaSlice request arrives.
			Therefore, we can define the decision epoch as an interval between the occurrence of two consecutive requests.
		
		\subsection{State Space}
		\label{subsec:state_space}
			Aiming to maximize revenue for the MISP with limited resources, several important factors need to be considered in the system state space.				
				First, the current system's available resources and the resources required by the current MetaSlice request are the two most important factors for the Admission Controller to decide whether it accepts this current request or not.
				Second, since the income for leasing each MetaSlice class is different, the class ID $i$ of a requested MetaSlice is another crucial information.
				Third, the similarity index $j$ of a requested MetaSlice reflects the similarity between a requested MetaSlice and the ongoing MetaInstances.
				Recall that the higher the value of $j$ is, the more similarity between the requested MetaSlice and the running MetaInstances is, leading to lower occupied resources when deploying this new request.
				Hence, the similarity index is also an important factor for the MetaSlice admission decision.
			
			We denote the available resources of the system by a vector \mbox{$\mathbf{n}_u\!=\![n_u^1, \dots, n_u^d,\dots, n_u^D]$} where $D$ is the total number of resource types, and $n_u^d$ denotes the number of available resources of type $d$.
			Similarly, the request's required resources are denoted by \mbox{$\mathbf{n}_m=[n_m^1, \dots, n_m^d,\dots, n_m^D]$)}, where $n_m^d$ is the number of requested resources of type $d$. 
			Given the above, the system state space can be defined as follows:
			\begin{equation}
				\label{state_space}
				\begin{aligned}
					\mathcal{S} \triangleq \Big\{\mbox{$\big(n_u^1, \dots, n_u^d,\dots, n_u^D,n_m^1, \dots, n_m^d,\dots, n_m^D, i,j$\big):}\\
					  \quad n_u^d \text{ and } n_m^d \in \{0,\dots,N^d\} \forall d \in \{1,\dots, D\}; \\
					  \quad i \in \{1,\dots,I\}; j \in [0,\dots, J]
					   \Big\},
				\end{aligned}	
			\end{equation}
				where $N^d$ is the maximum resources of type $d$, $I$ is the total number of classes in the system, and $J$ is the maximum similarity index of a MetaSlice derived by the MetaSlice Analyzer.
				By this design, the system state is presented by a tuple, i.e.,  $\mathbf{s} \triangleq (\mathbf{n}_u, \mathbf{n}_m, i, j)$, and the system can work continuously without ending at a terminal state, at which the system stops working~\cite{sutton2018reinforcement}.
			
			Recall that in the sMDP framework, the system only transits from state $\mathbf{s}$ to state $\mathbf{s}'$ if and only if an event occurs (e.g., a new MetaSlice request arrival).			
			We define the event as a vector \mbox{$\mathbf{e} \triangleq [e_1,\dots,e_i, \dots, e_I]$}, where \mbox{$e_i\in \{-1,0,1\}$}.
			Specifically, \mbox{$e_i\!=\!-1$} if a MetaSlice class $i$ departs from the system, \mbox{$e_i\!=\!1$} if a new MetaSlice request class-$i$ arises, and \mbox{$e_i\!=\!0$} otherwise (i.e., no MetaSlice request class-$i$ arrival or departure). 
			Thus, the set of all possible events is given as follows:
			\begin{equation}
				\mathcal{E} \triangleq \big\{\mathbf{e}:e_i\in \{-1,0,1\}; \sum_{i=1}^{I}|e_i| \leq 1  \big\}.
			\end{equation}
			Note that there is a trivial event $\mathbf{e^*} \triangleq (0,\dots, 0)$ meaning that no MetaSlice request of any class arrives or departs. 		
			 	
		\subsection{Action Space}
			If a MetaSlice request arrives at state $\mathbf{s}$, the  Admission Controller can decide whether to accept or reject this request to maximize the long-term revenue for the MISP.
			Thus, the action space at state $\mathbf{s}$ can be defined by:
			\begin{equation}
				\mathcal{A}_\mathbf{s} \triangleq \{0,1\}.
			\end{equation}			
			In particular, if the requested MetaSlice is accepted, the action at state $\mathbf{s}$ is equal to one, i.e., $a_\mathbf{s}=1$. Otherwise, $a_\mathbf{s}=0$. 
		\subsection{State Transition Probabilities}
			This sub-section analyzes the sMDP's dynamic by characterizing the underlying Markov chain's state transition probabilities.						
				Since the sMDP is based on the semi-Markov Process (sMP) that consists of a renewal process and a Continuous-time Markov Chain (CTMC) $\{X(t)\}$, the uniformization method can be used to derive the state transition probabilities $\mathcal{T}$~\cite{gallager_discrete_1995,tijms_a_2003,kallenberg_markov_online}.
				Specifically, the uniformization transforms the CTMC into a corresponding stochastic process $\{\bar{X}(t)\}$ whose transition epochs are derived from a Poisson process at a uniform rate, whereas state transitions follow a discrete-time Markov chain $\{X_n\}$.    
				These two processes, i.e., $\{X(t)\}$ and $\{\bar{X}(t)\}$, are proven to be probabilistically equivalent~\cite{tijms_a_2003}.
			
			In practice, similar to many communication systems, e.g., mobile phone systems, we have never known when a user request comes and leaves the system.
				Thus, we can consider that the arrival process of class-i requests follows the Poisson distribution with mean $\lambda_i$ while the departure of class-i MetaSlice follows an exponential distribution with mean $1/\mu_i$, as those in~\cite{gallager_discrete_1995}.
				In this way, the parameters of the uniformization method are defined as:
				\begin{align}
					z &= \max_{\mathbf{x}\in\mathcal{X}} \sum_{i=1}^{I}(\lambda_i + x_i\mu_i),\\
					z_{\mathbf{x}} &=  \sum_{i=1}^{I}(\lambda_i + x_i\mu_i),
				\end{align}					
				where each element of vector \mbox{$\mathbf{x} \triangleq [x_1,\dots,x_i,\dots,x_I]$} represents the number of on-going MetaSlices in the corresponding class (e.g., $x_i$ is the number of MetaSlices in class $i$ that are running simultaneously in the system), and $\mathcal{X}$ is the set containing all possible values of $\mathbf{x}$.
				Now, the events' probabilities are determined based on $z$ and $z_{\mathbf{x}}$ as follows:
				\begin{itemize}
					\item The probability of a class $i$ request  occurring in the next event $\mathbf{e}$ is $\lambda_i/z$.
					\item The probability of a class $i$ MetaSlice  departing in the next event $\mathbf{e}$ is $x_i\mu_i/z$.
					\item The probability of trivial event (i.e., no MetaSlice request of any class arrives or departs) arising in the next event $\mathbf{e}$ is $1-z_{\mathbf{x}}/z$.
				\end{itemize}
				Then, we can obtain the state transition probabilities \mbox{$\mathcal{T} =\{P_{\mathbf{s}\mathbf{s}'}(a_\mathbf{s})\}$} with \mbox{$\mathbf{s},\mathbf{s}' \in \mathcal{S}$} and \mbox{$a_\mathbf{s} \in \mathcal{A_\mathbf{s}}$}, i.e., the probability that the system moves between states by taking actions.
			
		\subsection{Immediate Reward Function}
		To maximize the MISP's long-term revenue, the immediate reward function needs to capture the income from leasing resources to Metaverse tenants.
			Here, we consider that the revenues for leasing resources for different classes are different since different classes may have different requirements such as reliability and delay.
			Recall that in our proposed Metaverse system, MetaSlices can share some functions with others, leading to differences in resource occupation even between MetaSlices from the same class.
			As such, even if two MetaSlices have the same income, accepting a MetaSlice that requires fewer resources will benefit the provider in the long term.
			Therefore, the number of resources required by a MetaSlice is another key factor.
			
			To this end, the immediate reward function can be defined as follows:
			\begin{align}
				r(\mathbf{s},a) = \left\{
				\begin{array}{ll}
					r_i -  \sum_{d=1}^{D}w_dn^d_o, &\mbox{if $e_i=1$ and $a = 1$},\\
					0, &\mbox{otherwise},					
				\end{array}	\right.
				\label{eq:reward_function}
			\end{align}
			where $r_i$ is the revenue from leasing resources for a MetaSlice class $i$, and $n^d_o$ is the number of type $d$ resources occupied by this requested MetaSlice. 
			The trade-offs between these factors are reflected by weights, i.e.,  \mbox{$\{w_i\}_{d=1}^D$}.				
			In practice, these weights are chosen by the MISP depending on its business strategies.
			Note that \eqref{eq:reward_function} implies that if slices have the same income (i.e., they are in the same class), accepting requests with fewer resource demands can help the provider to maximize the long-term revenue.
				
		\subsection{Optimization Formulation}
		Since the statistical characteristics (e.g., arrival rate and departure rate of a MetaSlice) of the proposed system are stationary (i.e., time-invariant), the policy $\pi$ for the meta-controller can be described as the time-invariant mapping from the state space to the action space, i.e., \mbox{$\pi:\mathcal{S} \rightarrow \mathcal{A}_\mathbf{s}$}.
			This study aims to find an optimal policy for the MetaSlicing's Admission Controller that maximizes a long-term average reward function $R_\pi(\mathbf{s})$, which is defined as an average expected reward obtained by starting from state $\mathbf{s}$ and following policy $\pi$ as follows:
			\begin{equation} 
				\label{eq:average_reward_function}
				\mathcal{R}_\pi(\mathbf{s})	=	\lim_{G \rightarrow \infty} \frac{\mathbb{E} \big[\sum_{g=0}^{G} r(\mathbf{s}_g,\pi(\mathbf{s}_g))|\mathbf{s}_0=\mathbf{s} \big]}{\mathbb{E} \big[\sum_{g=0}^{G} \tau_g |\mathbf{s}_0=\mathbf{s} \big]}, \forall \mathbf{s} \in \mathcal{S},	
			\end{equation}
			where $G$ is the total number of decision epochs, $\pi(\mathbf{s}_g)$ is the action derived by $\pi$ at decision epoch $g$, and $\tau_g$ is the interval time between two consecutive decision epochs. 
			The existence of the limit in $\mathcal{R}_\pi(\mathbf{s})$ is proven in Theorem~\ref{theorem1}.
			Thus, given the currently available resources and the MetaSlice request's information, the optimal policy $\pi^*$ can give optimal actions to maximize $\mathcal{R}_\pi(\mathbf{s})$, thereby maximizing the long-term revenue for the infrastructure provider. 			

		\begin{theorem}
			\label{theorem1}
			Given that the state space $\mathcal{S}$ and the number of decision epochs in a certain period of time are finite, we have:
			\begin{align}
				\mathcal{R}_\pi(\mathbf{s})	&=	\lim_{G \rightarrow \infty} \frac{\mathbb{E} \big[\sum_{g=0}^{G} r(\mathbf{s}_g,\pi(\mathbf{s}_g))|\mathbf{s}_0=\mathbf{s} \big]}{\mathbb{E} \big[\sum_{g=0}^{G} \tau_g |\mathbf{s}_0=\mathbf{s} \big]}\\
				&= \frac{\overline{\mathcal{T}}_\pi r(\mathbf{s}, \pi(\mathbf{s}))}{\overline{\mathcal{T}}_\pi y(\mathbf{s}, \pi(\mathbf{s}))} , \qquad \forall \mathbf{s} \in \mathcal{S},
			\end{align}
				where $r(\mathbf{s}, \pi(\mathbf{s}))$ is the expected immediate reward and $y(\mathbf{s}, \pi(\mathbf{s}))$ is the expected interval between two successive decision epochs when performing action \mbox{$\pi(\mathbf{s})$} at state $\mathbf{s}$, and $\overline{\mathcal{T}}_\pi$ is the limiting matrix of the embedded Markov chain corresponding to policy $\pi$, which is given based on the transition probability matrix of this chain, i.e., $\mathcal{T}_\pi$, as follows:
				\begin{equation}
					\label{eq:limiting_matrix}
					\overline{\mathcal{T}}_\pi = \lim_{G \rightarrow \infty} \frac{1}{G} \sum_{g=0}^{G-1} \mathcal{T}_\pi^g.
				\end{equation}
		\end{theorem}
		
		\begin{proof}
			The proof of Theorem~\ref{theorem1} is presented in Appendix~\ref{ap:proof_theorem1}.
		\end{proof}
		Note that the underlying Markov chain of our sMDP model is irreducible (i.e., the long-term average reward is independent of the starting state), which is proven in Theorem \ref{theorem2}.
		\begin{theorem}
			\label{theorem2}
			The long-term average reward $R_\pi(\mathbf{s})$ for any policy $\pi$ is well-defined and independent of the starting state, i.e., $R_\pi(\mathbf{s})=R_\pi, \forall \mathbf{s}\in \mathcal{S}$. 
		\end{theorem}
		\begin{proof}
		The proof of Theorem~\ref{theorem2} is presented in Appendix~\ref{ap:proof_theorem2}.
	\end{proof}
		Since the limiting matrix $\overline{\mathcal{T}}_\pi$ exists and the sum of probabilities that the system moves from one state to others is one, we have $\sum_{\mathbf{s'} \in \mathcal{S}}\overline{\mathcal{T}}_\pi(\mathbf{s}'|\mathbf{s})=1$.		
		Given the above, the MetaSlice admission control problem can be formulated as follows:
		\begin{equation} 
			\label{eq:optimization_formulation1}
			\begin{aligned}
				&\underset{\pi}{\text{max}} \qquad\qquad \left(\mathcal{R}_\pi	=	\frac{\overline{\mathcal{T}}_\pi r(\mathbf{s}, \pi(\mathbf{s}))}{\overline{\mathcal{T}}_\pi y(\mathbf{s}, \pi(\mathbf{s}))} \right),\\
				&\text{subject to:} \qquad \sum_{\mathbf{s'} \in \mathcal{S}}\overline{\mathcal{T}}_\pi(\mathbf{s}'|\mathbf{s})=1,\quad \forall \mathbf{s} \in \mathcal{S}.		
			\end{aligned}				
		\end{equation}			
		In the following section, we will discuss our proposed solution that can help the Admission Controller to obtain the optimal admission policy $\pi^*$ to maximize the long-term average reward function, i.e., \mbox{$\pi^* = \underset{\pi}{\text{argmax}}\mathcal{R}_\pi$}.
	\section{AI-based Solution with MetaSlice Analysis for MetaSlice Admission Management}
	\label{sec:solutions}
	This section presents our proposed approach for the MetaSlice admission management.
		We first discuss the main steps in the MetaSlice analysis to determine the similarity between a requested MetaSlice and ongoing MetaSlices.
		Then, we propose a Deep Reinforcement Learning (DRL)-based algorithm for the Admission Controller to address the real-time decision making requirement and the high uncertainty and dynamics of the request's arrival and MetaSlice departure processes, which are, in practice, often unknown in advance. 
		Thanks to the self-learning ability of DRL, the Admission Controller can gradually obtain an optimal admission policy via interactions with its surrounding environment without requiring complete knowledge of the arrival and departure processes of MetaSlices in advance.
		
		\subsection{MetaSlice Analysis}
		\label{subsec:meta_analysis}
		Recall that the major role of the MetaSlice Analyzer is to analyze the request's MetaBlueprint to determine the similarity between the requested MetaSlice and the ongoing MetaInstances.
			Then, the similarity report is used to assist the Admission Controller in deciding whether to accept or reject a request.		
			This paper uses the function configuration to decide the similarity index since it clearly shows the relationship between the requested MetaSlice and the ongoing MetaInstance.
		
		We consider that the proposed framework supports $F$ types of functions.
			Then, we can denote the function configuration of a MetaSlice $m$ by a set $\mathcal{F}_m$ as follows: 
		\begin{equation}
			\mathcal{F}_m=\left\{\mathbf{f}_1^m,\dots, \mathbf{f}_f^m,\dots, \mathbf{f}_F^m\right\},
		\end{equation}
			where the configuration of function $f$ is represented by a vector with size $K$, i.e., \mbox{$\mathbf{f}_f \in \{0,1\}^K$}.
			We define a trivial function configuration vector \mbox{$\mathbf{f}^* \triangleq (0,\dots,0)$} meaning that a function is not required by the MetaSlice. 
			Note that a MetaInstance also maintains a function configuration set and updates it whenever the MetaSlice is initialized or released.
		
		Given two function configuration sets $\mathcal{F}_1$ and $\mathcal{F}_2$, the similarity index can be given as follows:
			\begin{equation}
				j(\mathcal{F}_1,\mathcal{F}_2) = \frac{1}{F}\sum_{f=1}^{F}b(\mathbf{f}_f^1,\mathbf{f}_f^2),
			\end{equation}
			where $b$ can be any similarity function used for vectors such as Jaccard and Cosine similarity functions\cite{jaccard1912the}.
			Here, we use Jaccard similarity that is defined as follows:
			\begin{equation}
				b_\text{Jaccard}(\mathbf{f}_f^1,\mathbf{f}_f^2) = \frac{\mathbf{f}_f^1\cdot\mathbf{f}_f^2}{||\mathbf{f}_f^1||^2 +||\mathbf{f}_f^2||^2 - \mathbf{f}_f^1\cdot\mathbf{f}_f^2},
			\end{equation}	
			where the numerator is the dot product of two vectors and $||\cdot||$ is the Euclidean norm of the vector, which is calculated as the square root of the sum of the squared vector's elements.
					
		The similarity index plays two roles in the MetaSlicing framework.
		First, it provides the Admission Controller with precious information for making decisions.
		Second, based on the accepted request's similarity index, MetaSlicing's resource management determines to put it in an existing MetaInstance or create a new one.   
		
		\subsection{Deep Dueling Double Q-learning based-Admission Controller}
		\label{subsec:D3QL}
		In Reinforcement Learning (RL), Q-learning is widely adopted due to its simplicity in implementation and convergence guarantee after the learning phase~\cite{watkins_q_1992}.
			Nevertheless, using a table to estimate the optimal values of all state-action pairs $Q^*(\mathbf{s},a)$, i.e., Q-values, hinders the Q-learning from being applied in a high-dimensional state space as the problem considered in this paper with hundred thousand states~\cite{luong_application_2019}. 
			In addition, the usage of the Q-table is only feasible when values of states are discrete, but the similarity score in the considered state can be a real number.
			These challenges are addressed by deep Q-learning, in which a Deep Neural Network (DNN), instead of a Q-table, is used  to approximate $Q^*(\mathbf{s},a)$ for all state-action pairs~\cite{mnih_human_2015}.
			However, both Q-learning and deep Q-learning have the same problem of overestimation when estimating Q-values~\cite{hasselt_doubledeep_2016}. 
			This issue makes the learning process unstable or even results in a sub-optimal policy if overestimations are not evenly distributed across states~\cite{thrun_issues_1993}.
		\begin{figure}[t]
			\centering
			\includegraphics[width=0.8\linewidth]{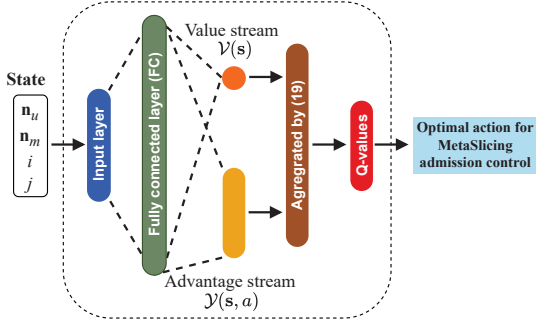}
			\caption{The architecture of iMSAC's DNN}
			\label{fig:dueling_architecture}
		\end{figure}
		To that end, this work proposes a deep RL algorithm for the MetaSlicing Admission Controller, namely iMSAC, that can address these above issues effectively by leveraging three innovative techniques: (i) the memory replay mechanism, (ii) the dueling neural network architecture, and (iii) the deep double Q-learning (DDQL). 
			The details of iMSAC are presented in Algorithm~\ref{alg:d3ql}. 
			The rest of this subsection will explain the key components of the iMSAC. 
			
		Given that the learning phase of the proposed algorithm consists of $T$ time steps, and at time step $t$, the Admission Controller observes the current state $\mathbf{s}_t$ and takes an action $a_t$ according to the $\epsilon$-policy. 
			After that, it observes a next state $\mathbf{s}_{t+1}$ and gets a reward $r_t$.
			This experience data, represented by a tuple \mbox{$<\mathbf{s}_t, a_t, \mathbf{s}_{t+1}, r_t>$}, cannot be used directly to train the DNN since the consecutive experiences are highly correlated, which may lead to a slow convergence rate~\cite{halkjaer_effect_1996}.
			As such, we adopt the memory replay mechanism where experiences are stored in a buffer \textbf{B}, as illustrated in Fig.~\ref{fig:d3ql_model}. 
			Then, at each time step, experiences are sampled uniformly at random to train the DNN.
			By doing so, correlations among experiences can be removed, thereby accelerating the learning process.
			Moreover, as one data point can be used multiple times to train the DNN, this mechanism can indeed improve the data usage efficiency.   
		\begin{figure}[t]
			\centering
			\includegraphics[width=0.95\linewidth]{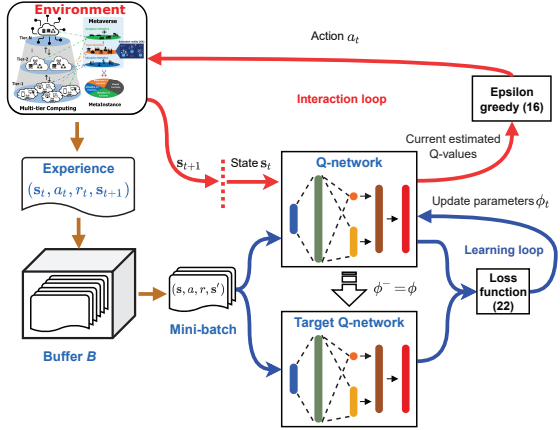}
			\caption{The proposed iMSAC-based Admission Controller for the MetaSlicing framework.}
			\label{fig:d3ql_model}
		\end{figure}
		\begin{algorithm}[t]
			\caption{The iMSAC}
			\label{alg:d3ql}
			\begin{algorithmic}[1]
				\STATE Initialize $\epsilon$ and buffer $\mathbf{B}$. 
				\STATE Create Q-network $\mathcal{Q}$ with random parameters $\phi$. 
				\STATE Create target Q-network $\hat{\mathcal{Q}}$ by cloning the Q-network.
				\FOR{\textit{step = 1 to T}}
				\STATE Get action $a_t$ following the $\epsilon$-greedy policy as follows:
					\begin{align}
						a_t\! =\! \left\{
						\begin{array}{ll}						
							\!\underset{a \in \mathcal{A}}{\text{argmax}}~\mathcal{Q}(\mathbf{s}_t,a; \phi_t),&\mbox{with probability }1-\epsilon,\\
							\!\mbox{random action }a\!\in\!\mathcal{A}, &\mbox{otherwise}.												
						\end{array}	\right.
						\label{eq:epsilon_greedy}
					\end{align}
				\STATE Perform $a_t$, then observe reward $r_t$ and next state $\mathbf{s}_{t+1}$.
				\STATE Store experience $(\mathbf{s}_t, a_t, r_t, \mathbf{s}_{t+1})$ in  $\mathbf{B}$.
				\STATE Create a mini-batch of experiences by sampling randomly from buffer $\mathbf{B}$, i.e., $(\mathbf{s}, a, r, \mathbf{s}') \sim U( \mathbf{B}$).				
				\STATE Obtain the Q-value and the target Q-value by using~\eqref{eq:recontruct_Qfunction_mean} and~\eqref{eq:targetDQN}, respectively.
				\STATE Update $\phi$ based on Stochastic Gradient Descent (SGD) algorithm. 
				\STATE Decrease the value of $\epsilon$.
				\STATE Set $\phi^-= \phi$ at every $C$ steps.
				\ENDFOR
			\end{algorithmic}
		\end{algorithm}
		
		In the proposed iMSAC, since a DNN is used to approximate the Q-values, we define the input and output layers of the DNN according to the state dimensions (i.e., available system resources, requested MetaSlice's resources, class ID, and similarity score) and action space.
			Specifically, when feeding a state to the DNN, we obtain the Q-values for all actions at this state, each given by a neuron at the DNN's output layer. 
			To improve the stability and increase the convergence rate, we propose to use the state-of-the-art dueling neural network architecture~\cite{wang_dueling_2016} for the iMSAC's DNN, as shown in Fig.~\ref{fig:dueling_architecture}.
			In particular, the dueling architecture divides the DNN into two streams.
			The first one estimates the state-value function $\mathcal{V}(\mathbf{s})$, which indicates the value of being at a state $\mathbf{s}$.
			The second stream estimates the advantage function $\mathcal{Y}(\mathbf{s}, a)$ that demonstrates the importance of action $a$ compared to other actions at state $\mathbf{s}$.
				
		Note that the state-action value function $\mathcal{Q}(\mathbf{s},a)$ (namely Q-function) expresses the value of taking an action $a$ at state $\mathbf{s}$, i.e., Q-value.
			Thus, the advantage function under policy $\pi$ can be given as 
			$\mathcal{Y}^\pi(\mathbf{s},a) = \mathcal{Q}^\pi(\mathbf{s},a) - \mathcal{V}^\pi(\mathbf{s})$~\cite{wang_dueling_2016}.
			Then, we can obtain the estimated Q-function for feeding state $\mathbf{s}$ to the \textit{i-AC}'s Deep Dueling Neural Network (DDNN) as follows:
			\begin{equation}
				\label{eq:recontruct_Qfunction}
				\mathcal{Q}(\mathbf{s},a; \zeta,\beta) = \mathcal{V}(\mathbf{s}; \zeta) + \mathcal{Y}(\mathbf{s},a; \beta),		
			\end{equation}
			where $\zeta$ and $\beta$ are the parameters of the state-value and advantage streams, respectively.
			It can be observed that given $\mathcal{Q}$, $\mathcal{V}$ and $\mathcal{Y}$ could not be determined uniquely.
			For instance, $\mathcal{Q}$ is unchanged if $\mathcal{Y}$ decreases the same amount that $\mathcal{V}$ increases.  
			As such, using \eqref{eq:recontruct_Qfunction} directly may result in a poor performance of the algorithm.
			Therefore, similar to~\cite{wang_dueling_2016}, we propose to use the following output of the advantage stream:
			\begin{equation}
				\label{eq:recontruct_Qfunction_max}
				\mathcal{Q}(\mathbf{s},a; \zeta,\beta) = \mathcal{V}(\mathbf{s}; \zeta) + \Big(\mathcal{Y}(\mathbf{s},a; \beta) - \max_{a' \in \mathcal{A}_\mathbf{s}} \mathcal{Y}(\mathbf{s},a';\beta)\Big).		
			\end{equation}
		In this way, for the optimal action $a^*$ at state $\mathbf{s}$, i.e., 
			\[a^* =  \underset{a \in \mathcal{A}_\mathbf{s}}{\text{argmax}}~\mathcal{Q}(\mathbf{s},a; \zeta,\beta) = \underset{a \in \mathcal{A}_\mathbf{s}}{\text{argmax}}~\mathcal{Y}(\mathbf{s},a; \beta),\]
			the 	 
			$\mathcal{Q}(\mathbf{s},a^*; \zeta,\beta)$ is forced to equal $\mathcal{V}(\mathbf{s},\zeta)$.
			However, \eqref{eq:recontruct_Qfunction_max} still faces an issue, i.e., the advantage function changes at the same speed at which the advantage of the predicted optimal action changes, making the estimation of Q-values unstable.
			To address this issue, the max operator is replaced by the mean as follows~\cite{wang_dueling_2016}: 
			\begin{equation}
				\label{eq:recontruct_Qfunction_mean}
				\mathcal{Q}(\mathbf{s},a; \zeta,\beta) = \mathcal{V}(\mathbf{s}; \zeta) + \Big(\mathcal{Y}(\mathbf{s},a; \beta) - \frac{1}{|\mathcal{A}_\mathbf{s}|} \sum_{a'\in\mathcal{A}_\mathbf{s}}\mathcal{Y}(\mathbf{s},a';\beta)\Big).		
			\end{equation}
		
		The root of the overestimation problem in Q-learning and deep Q-learning comes from the max operation when estimating the target Q-value at time $t$ as follows~\cite{hasselt_doubledeep_2016}:
		\begin{equation}
			\label{eq:estimatin_Qvalues}
			\begin{aligned}
				H_t = r_t(\mathbf{s}_t, a_t) + \gamma\max_{a_{t+1}} \mathcal{Q}(\mathbf{s}_{t+1}, a_{t+1}),
			\end{aligned}
		\end{equation}
		where $\gamma$ is the discount factor indicating the importance of future rewards. 
		To handle this issue, we adopt the deep double Q-learning algorithm~\cite{hasselt_doubledeep_2016} that leverages two identical deep dueling neural networks.
			One deep neural network is for action selection, namely Q-network $\mathcal{Q}$, and the other is for action evaluation, namely target Q-network $\hat{\mathcal{Q}}$.
			Then, the target Q-value is computed as follows:
			\begin{equation}
				\label{eq:targetDQN}
				H_t = r_t + \gamma \hat{\mathcal{Q}}\big(\mathbf{s}_{t+1}, \underset{a}{\operatorname{argmax}}\mathcal{Q}(\mathbf{s}_{t+1},a; \phi_t, );\phi^-_t\big),
			\end{equation}
			where $\phi$ and $\phi^-$ are the parameters of the Q-network and target Q-network, respectively.
		
		As the aim of training the Q-network is to minimize the gap between the target Q-value and the current estimated Q-value, the loss function at time $t$ is given as follows:
		\begin{equation}
			\begin{aligned}
				\label{eq:lossfunction}
				\mathcal{L}_t(\phi_t) = \mathbb{E}_{(\mathbf{s},a,r,\mathbf{s}')}\bigg[ \bigg( H_t
				-\mathcal{Q}(\mathbf{s},a;\phi_t)\bigg)^2\bigg],
			\end{aligned}
		\end{equation}
			where $\mathbb{E}[\cdot]$ is the expectation according to a data point \mbox{$(\mathbf{s},a,r,\mathbf{s}')$} in the buffer \textbf{B}.
			The loss function $\mathcal{L}_t(\phi_t)$ can be minimized by the Gradient Descent (GD), which is one of the most popular algorithms for minimizing deep learning loss functions due to its simplicity in implementation~\cite{du2019gradient}.
			However, GD requires calculating the gradient and cost function for all data points at each time step, yielding a very high processing time for large data.
		
	To that end, we propose to use the Stochastic Gradient Descent (SGD) that only needs to compute gradients and cost function for a mini-batch data, i.e., $\mathcal{D}$, sampled uniformly from the buffer \textbf{B}.
		By doing so, SGD can accelerate the convergence  rate while still guaranteeing the convergence of learning~\cite{robbins1951stochastic}. 
		Specifically, the cost function of SGD is defined as follows:
		\begin{equation}
			\begin{aligned}
				\label{eq:costfunction}
				J_t(\phi_t) = \frac{1}{|\mathcal{D}|}\sum_{(\mathbf{s},a,r,\mathbf{s}') \in \mathcal{D}}\mathcal{L}_t(\phi_t).
			\end{aligned}
		\end{equation}
		At each step, Q-network's parameters are updated as follows:
		\begin{equation}
			\label{eq:GDupdate}
			\phi_{t+1} = \phi_t -\alpha_t \nabla_{\phi_t} J_t(\phi_t),
		\end{equation}  
		where $\nabla_{\phi_t}(\cdot)$ is the gradient of the cost function according to the current Q-network's parameters $\phi_t$ and $\alpha_t$ is the step size controlling the rate of updating the Q-network's parameters.
		It is worth mentioning that even though the target Q-value $H_t$ in~\eqref{eq:lossfunction} looks like labels that are used in the supervised learning, $H_t$ is not fixed before starting the learning process.
		Moreover, it changes at the same rate as that of the target Q-network's parameters, possibly leading to instability in the learning process.
		To that end, instead of updating the parameters of $\hat{\mathcal{Q}}$ at every time step, $\phi^-$ is only updated by copying from $\phi$ at every $U$ steps.  
	
	The computational complexity of the proposed iMSAC algorithm is mainly determined by the training phase of the Q-network. 
		In this DNN, input and hidden layers have $Y$ and $L$ neurons, respectively. The output layer consists of the value stream with $B$ neurons and the advantage stream with $C$ neurons.
		Since training the Q-network at each iteration is implemented with matrix multiplication~\cite{goodfellow2016deep}, the computational complexity of iMSAC for training with a data-point is $\mathcal{O}(YL + LB + LC)$.
		Given the training phase takes $T$ iterations, each with a batch of experiences with size $Z$, the computational complexity of iMSAC is $\mathcal{O}(TZ(YL + LB + LC))$.

\section{PERFORMANCE EVALUATION}
	\label{sec:results}
	\subsection{Simulation Parameters}
	\label{subsec:parameters}
	The parameters for our simulation are set as follows, unless otherwise stated.
		We consider that the system supports up to nine types of functions, i.e., \mbox{$F\!=\!9$}.
		Each MetaSlice consists of three different functions and belongs to one of three classes, i.e., class-1, class-2, and class-3.	
		In the configuration set $\mathcal{F}$, we set \mbox{$K\!=\!1$}.	
		Here, we set $\lambda_1$, $\lambda_2$, and $\lambda_3$ to $60$, $40$, and $25$ requests/hour, respectively, and its vector is denoted by \mbox{$\boldsymbol{\lambda} = [60,40,25]$}.
		The average MetaSlice session time is $30$ minutes, i.e, $\mu_i=2,~\forall i\in \{1,2,3\}$.
		The immediate reward $r_i$ for accepting a request from class-1, class-2, and class-3 are 1, 2, and 4, respectively.
		Note that our proposed algorithm, i.e., iMSAC, does not require the above information in advance.
		It can adjust the admission policy according to the practical requirements and demands (e.g., rental fees, arrival rate, and total resources of the system) to maximize a long-term average reward.
		Therefore, without loss of generality, the system has three types of resources, i.e., computing, storage, and radio bandwidth.
		Each Metaverse function is assumed to require a similar amount of resources as those of the Network Slice in 5G Network Slicing~\cite{ghina_an_2019}, e.g., $40$ GB for storage, a bandwidth of $40$ MHz, and $40$ GFLOPS/s for computing.	

	In our proposed algorithm, i.e., iMSAC, the settings are as follows.
		For the $\epsilon$-greedy policy, the value of $\epsilon$ is gradually decreased from $1$ to $0.01$.
		The discount factor $\gamma$ is set to $0.9$.
		We use Pytorch to build the Q-network and the target Q-network.
		They have the same architecture as shown in Fig.~\ref{fig:dueling_architecture}.
		During the learning process, typical hyperparameters of DNN are selected as those in~\cite{mnih_human_2015} and \cite{hasselt_doubledeep_2016}, e.g., the learning rate of the Q-network is set at $10^{-3}$ and the target-Q network's parameters are copied from the parameters of Q-network at every $10^4$ steps. 	
			
	Recall that our proposed solution consists of two important elements, i.e., the intelligent algorithm iMSAC and MetaSlice analysis with the MetaInstance technique (MiT). 
		With MetaInstance, functions can be reused, leading to a significant improvement in resource utilization.
		Meanwhile, the iMSAC can help the Admission Controller to obtain an optimal policy without requiring the complete information about the arrival and departure of MetaSlice in advance.
		Therefore, we compare our proposed solution, i.e., iMSAC+MiT, with three counterpart approaches: (i) iMSAC, (ii) Greedy policy~\cite{sutton2018reinforcement} where the MetaSlicing's Admission Controller accepts a request if the system has enough resources for the request, and (iii) Greedy policy with the MetaInstance technique, i.e., Greedy+MiT.		
	\begin{figure}[t]
		\centering
		\includegraphics[width=0.8\linewidth]{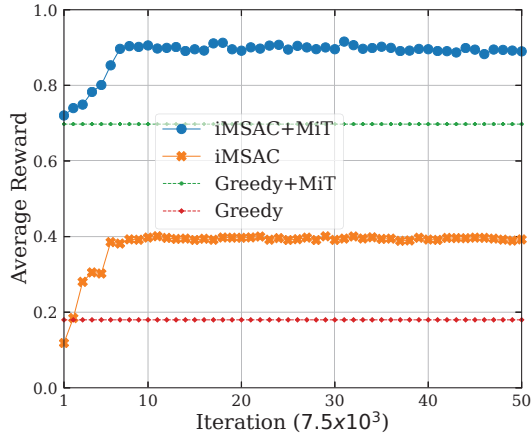} 
		\caption{Convergence rate of iMSAC.}
		\label{fig:convergene}
	\end{figure}
	\begin{figure*}[t]
		\centering
		$\begin{array}{cccc}
			\begin{array}{ccc}
				\includegraphics[width=0.25\linewidth]{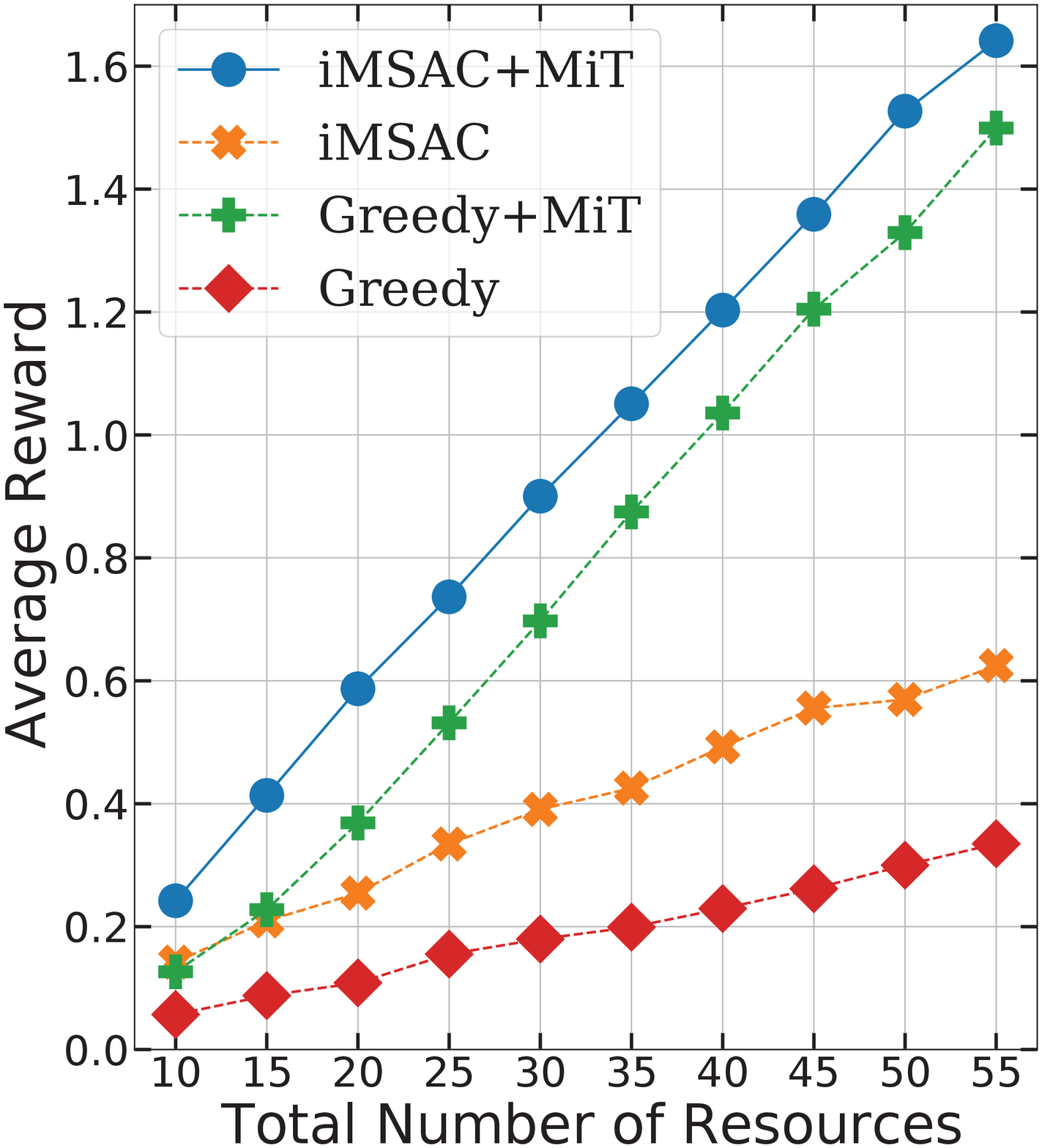}
				&\includegraphics[width=0.25\linewidth]{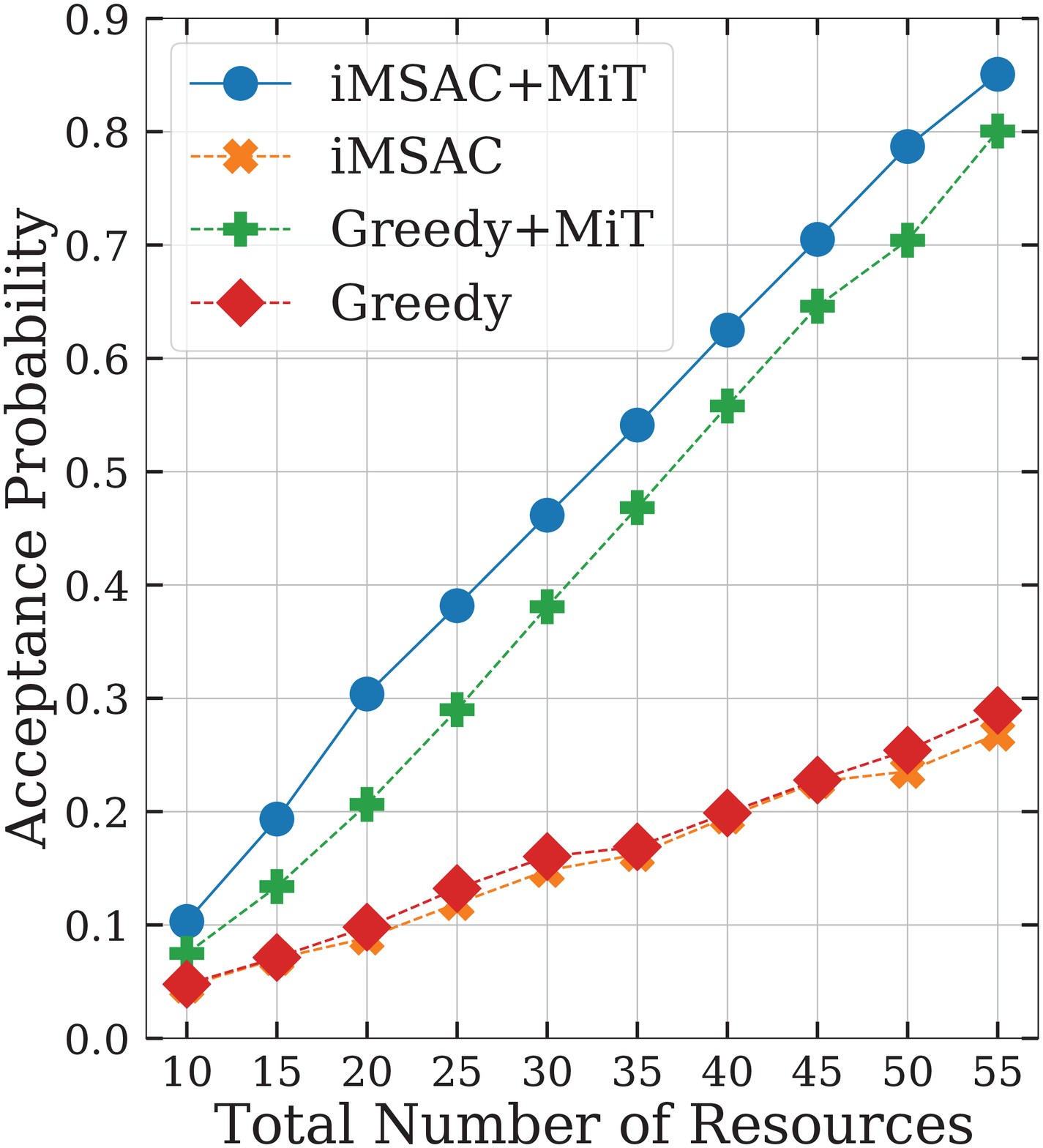} 
				&\includegraphics[width=0.25\linewidth]{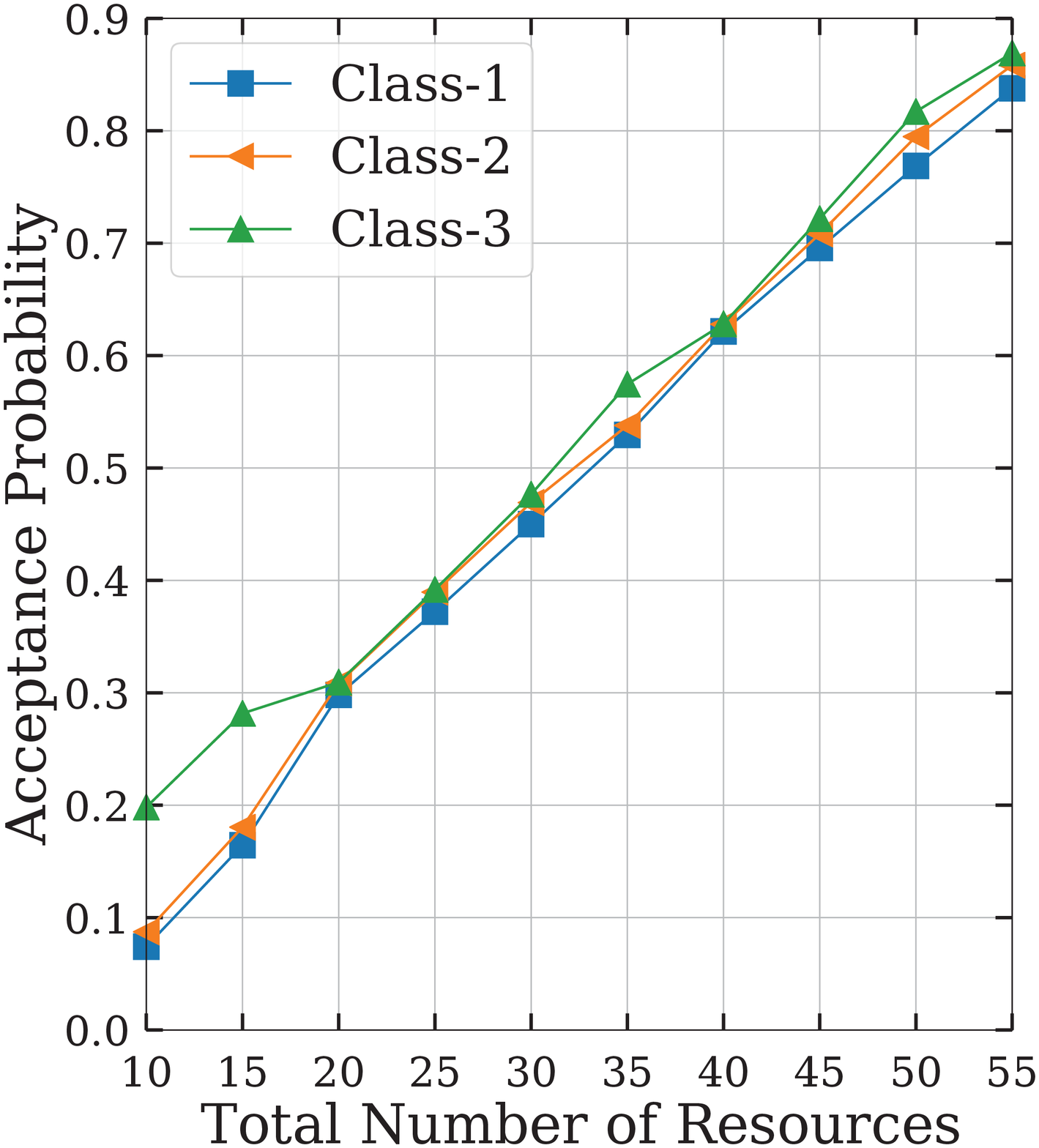}\\
				\text{(a) Average rewards}&\text{(b) Acceptance probability}&\text{(c) iMSAC+MiT}	
			\end{array}\\	
			\begin{array}{ccc}
				\includegraphics[width=0.25\linewidth]{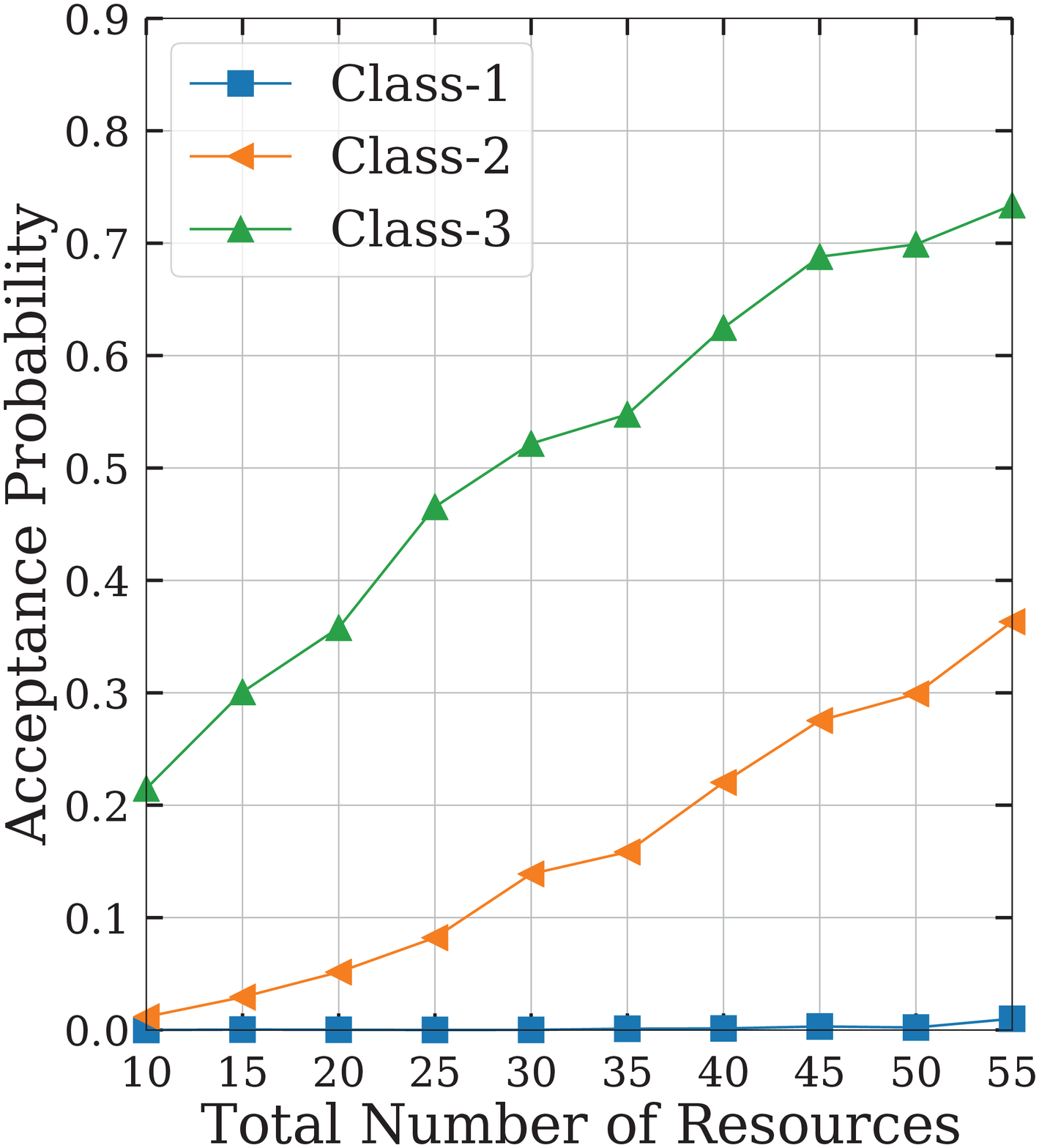}
				&\includegraphics[width=0.25\linewidth]{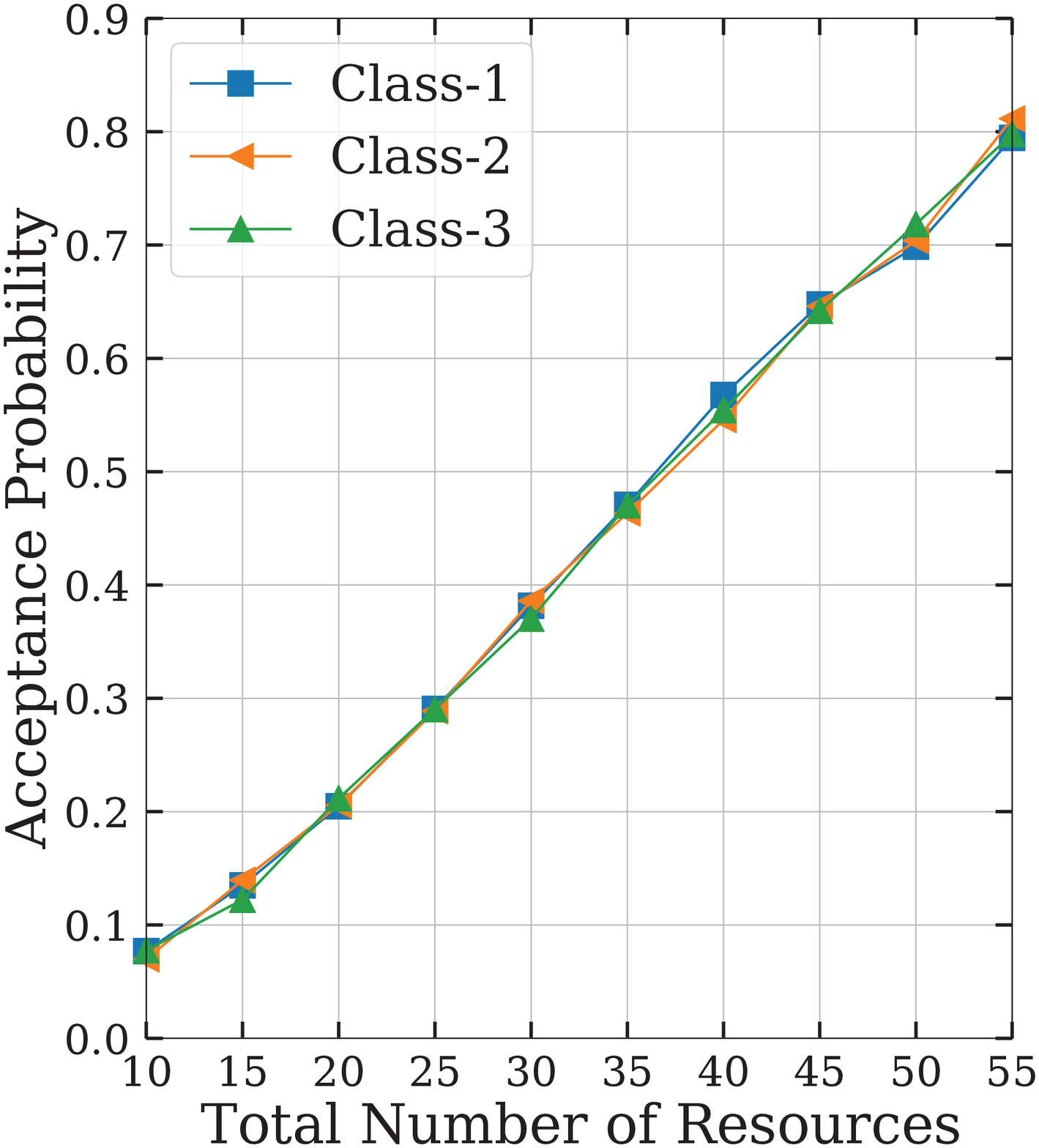} 
				&\includegraphics[width=0.25\linewidth]{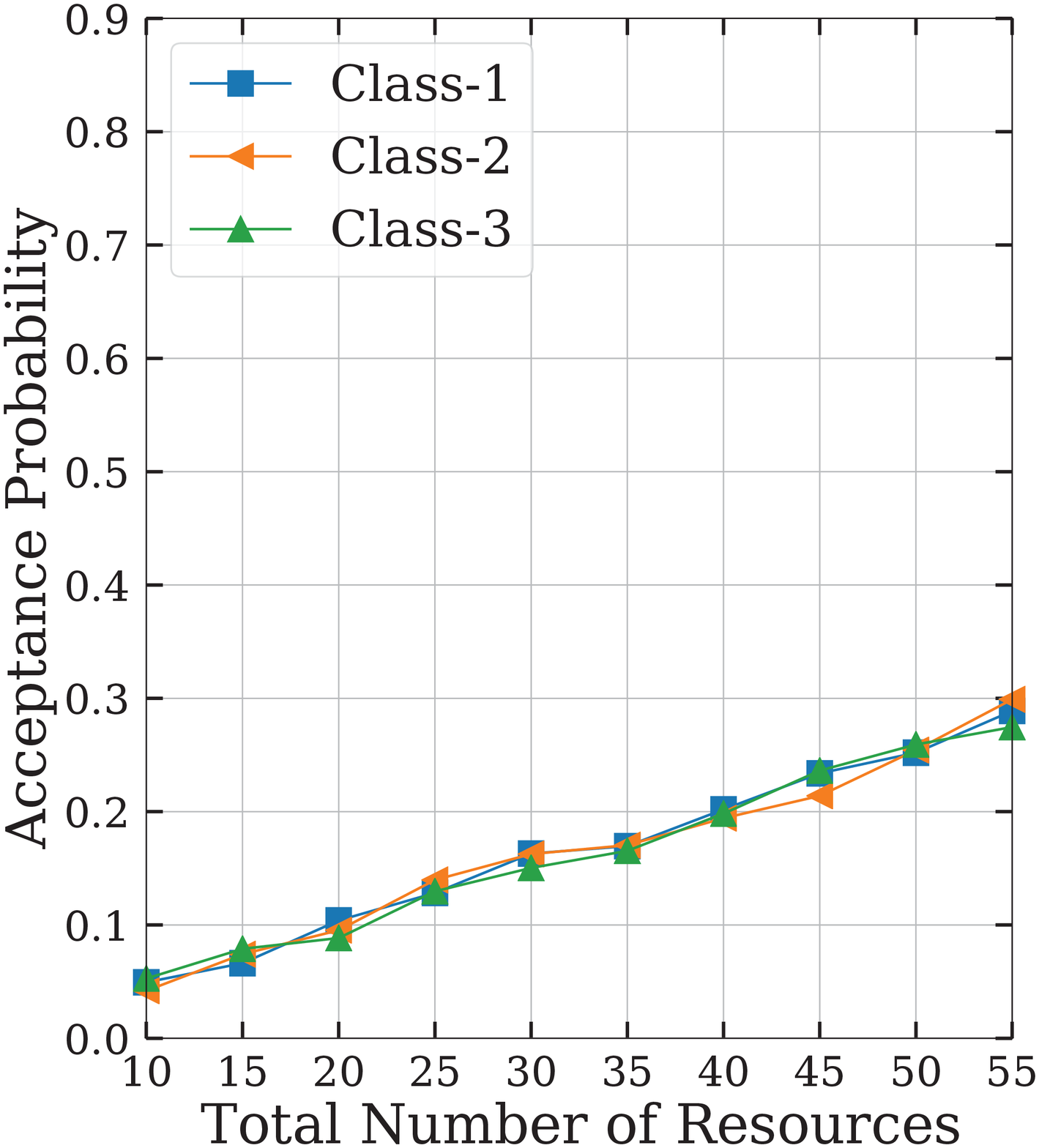}\\
				\text{(d) iMSAC} &\text{(e) Greedy+MiT} &\text{(f) Greedy}
			\end{array}
		\end{array}$
		\caption{Vary the total number of system resources.}
		\label{fig:vary_vms}
	\end{figure*}

	\subsection{Simulation Results}
	Simulations are conducted to gain insights into our proposed solution, i.e., iMSAC+MiT.
		First, we will investigate the convergence rate of our proposed algorithm iMSAC.
		Then, we evaluate the proposed solution in different scenarios to study impacts of important system parameters, e.g., the available system resources, immediate rewards reflecting the revenue of the MISP, and the maximum number of MetaSlices sharing one function that is one of the most important parameters of the MISP.
	
	\subsubsection{Convergence Rate} 
	\label{subsubsec:cvg}

	Figure~\ref{fig:convergene} shows the convergence rates of our proposed iMSAC algorithm in two scenarios with and without the MetaInstance technique. 
		In this experiment, we set storage, radio bandwidth, and computing resources to $1200$ GB, $1200$ MHz, and $1200$ GFLOPS/s, respectively.
		In other words, the system can support up to $30$ functions in total.
		The average rewards obtained by Greedy+MiT and Greedy are also presented for comparisons. 
		Specifically, the learning curves of iMSAC+MiT and iMSAC have a very similar trend.
		As shown in Fig.~\ref{fig:convergene}, both of them gradually converge to the optimal policy after~\mbox{$6\!\times\!10^4$} iterations.
		However, the iMSAC+MiT's average reward is stable at $0.9$, which is $2.25$ times greater than that of the iMSAC.
		Similarly, the Greedy+MiT's average reward is much greater (i.e., $3.5$ times) than that of the Greedy. 
		Thus, these results clearly show the benefits of the iMSAC and the MetaInstance technique.
		In particular, while the MetaInstance can help to maximize the resource utilization for the system, the iMSAC can make the Admission Controller learn the optimal policy to maximize the long term average reward.

	\subsubsection{Performance Evaluation}
	\begin{figure*}[t]
		\centering
		$\begin{array}{ccc}
			\includegraphics[width=0.25\linewidth]{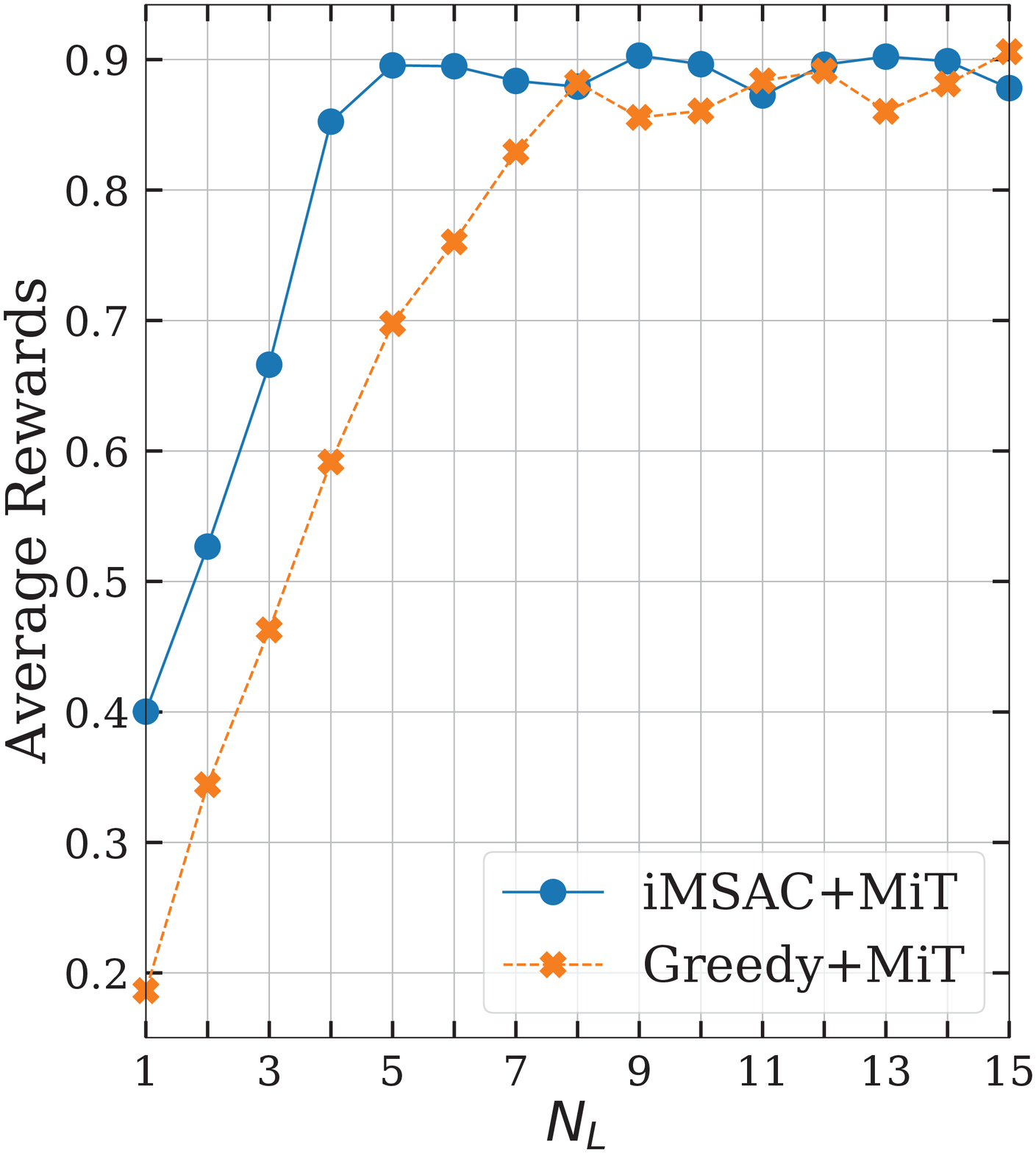}
			&\includegraphics[width=0.25\linewidth]{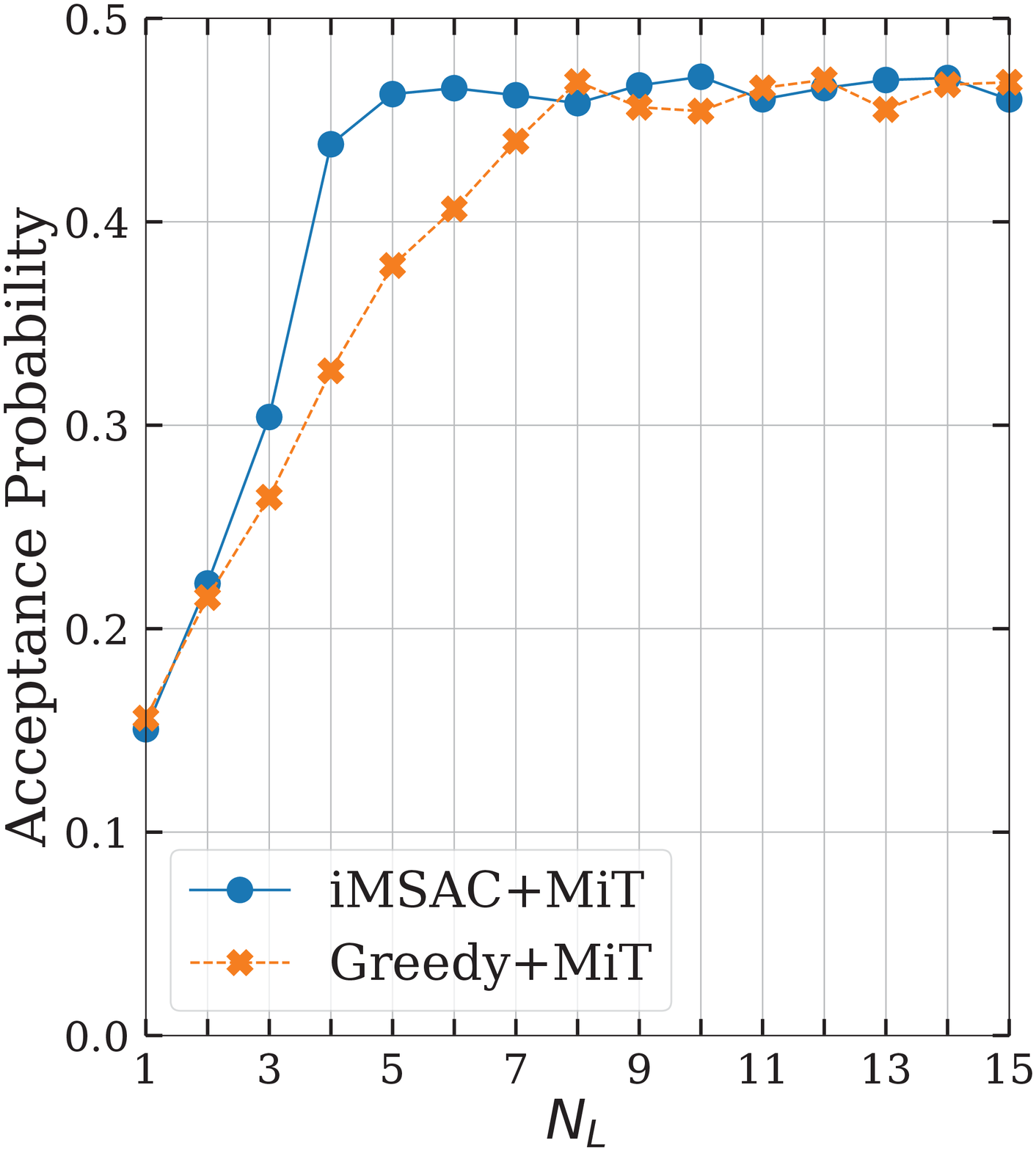} 
			&\includegraphics[width=0.25\linewidth]{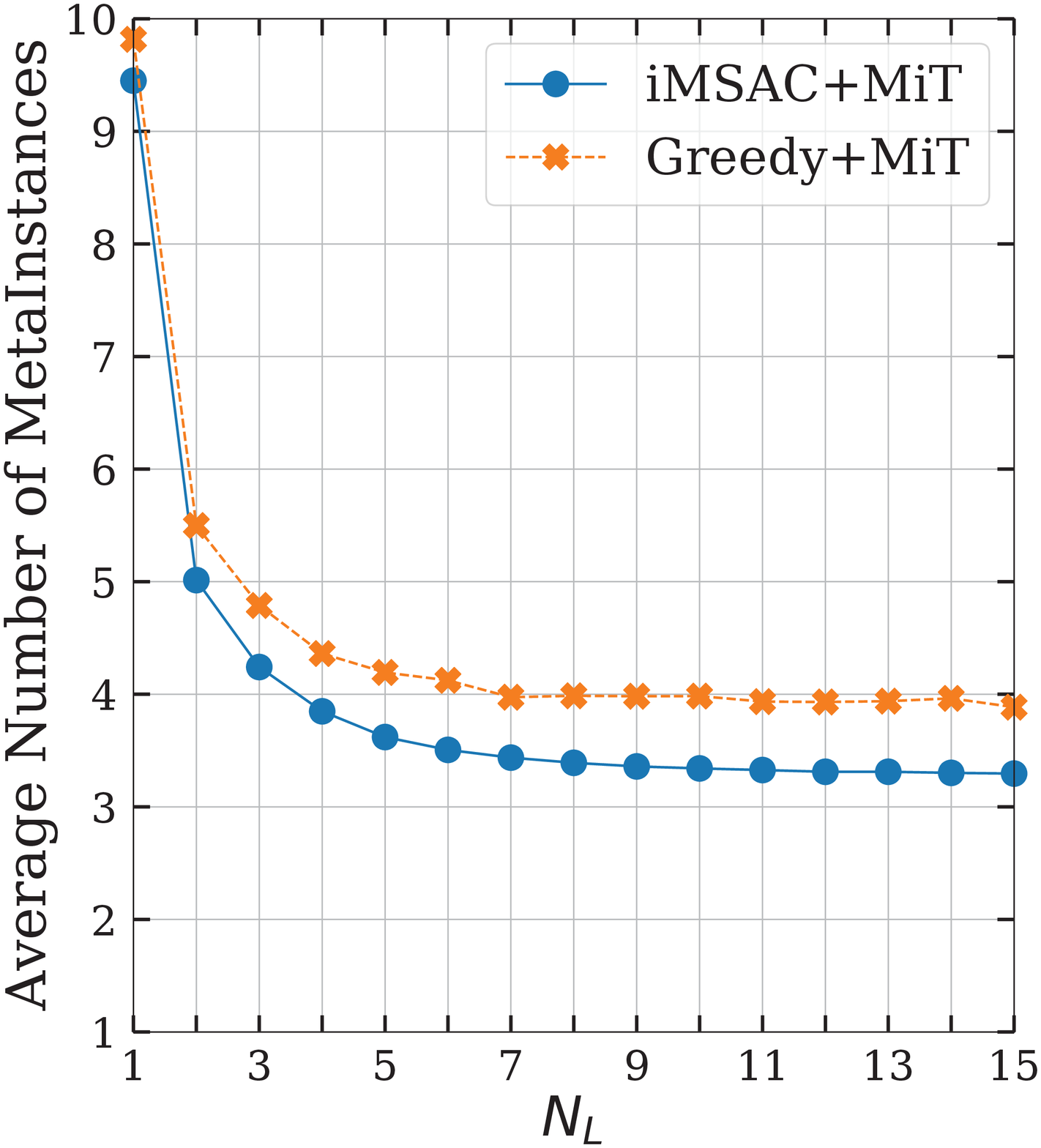}\\
			\text{(a) Average rewards}&\text{(b) Acceptance probability}&\text{(c) Average number of MetaInstances}
		\end{array}$
		\caption{Varying $N_L$}
		\label{fig:vary_app_limit}
	\end{figure*}
	
	\begin{figure*}[t]
		\centering
		$\begin{array}{cccc}
			\includegraphics[width=0.25\linewidth]{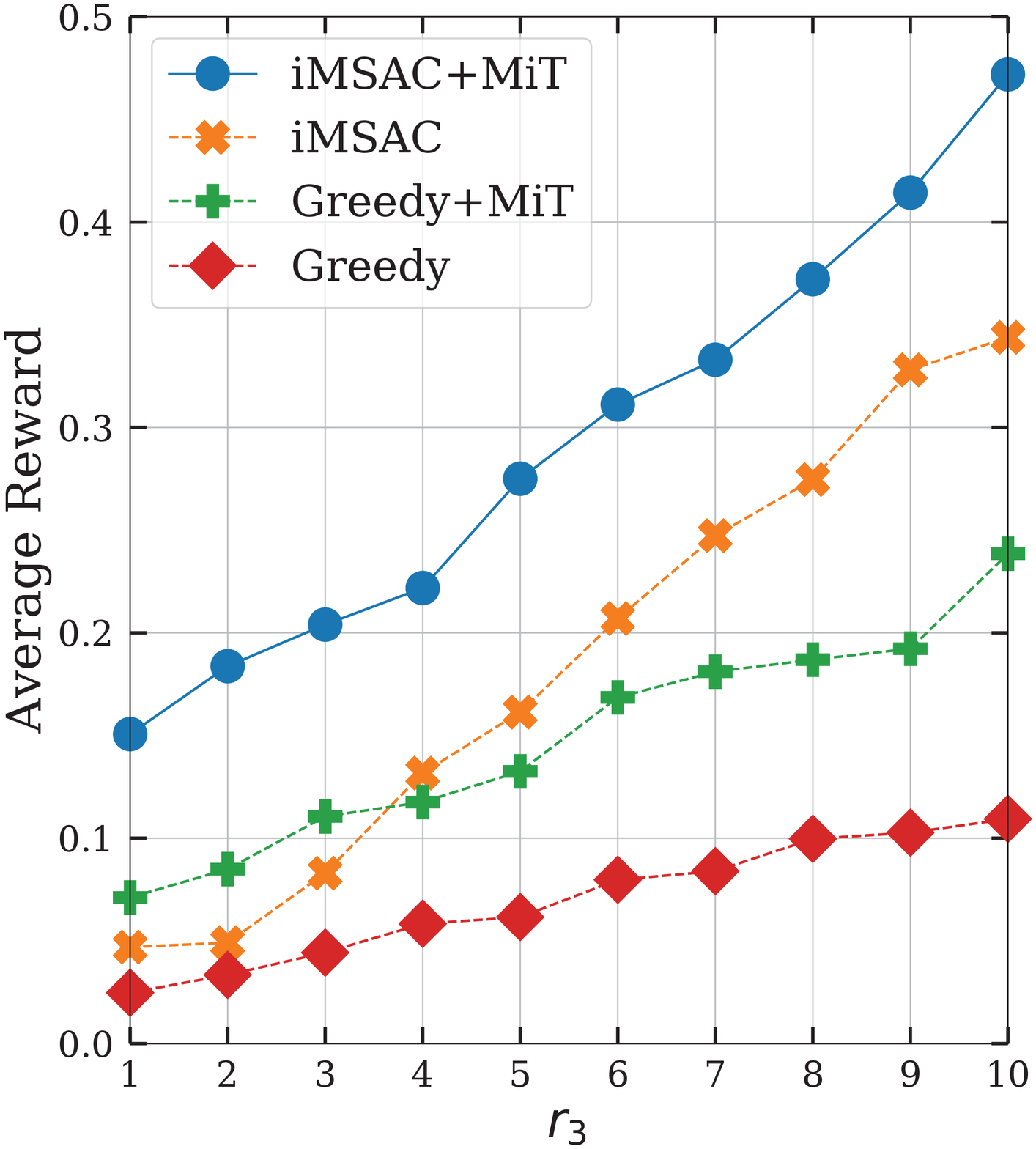}
			&\includegraphics[width=0.25\linewidth]{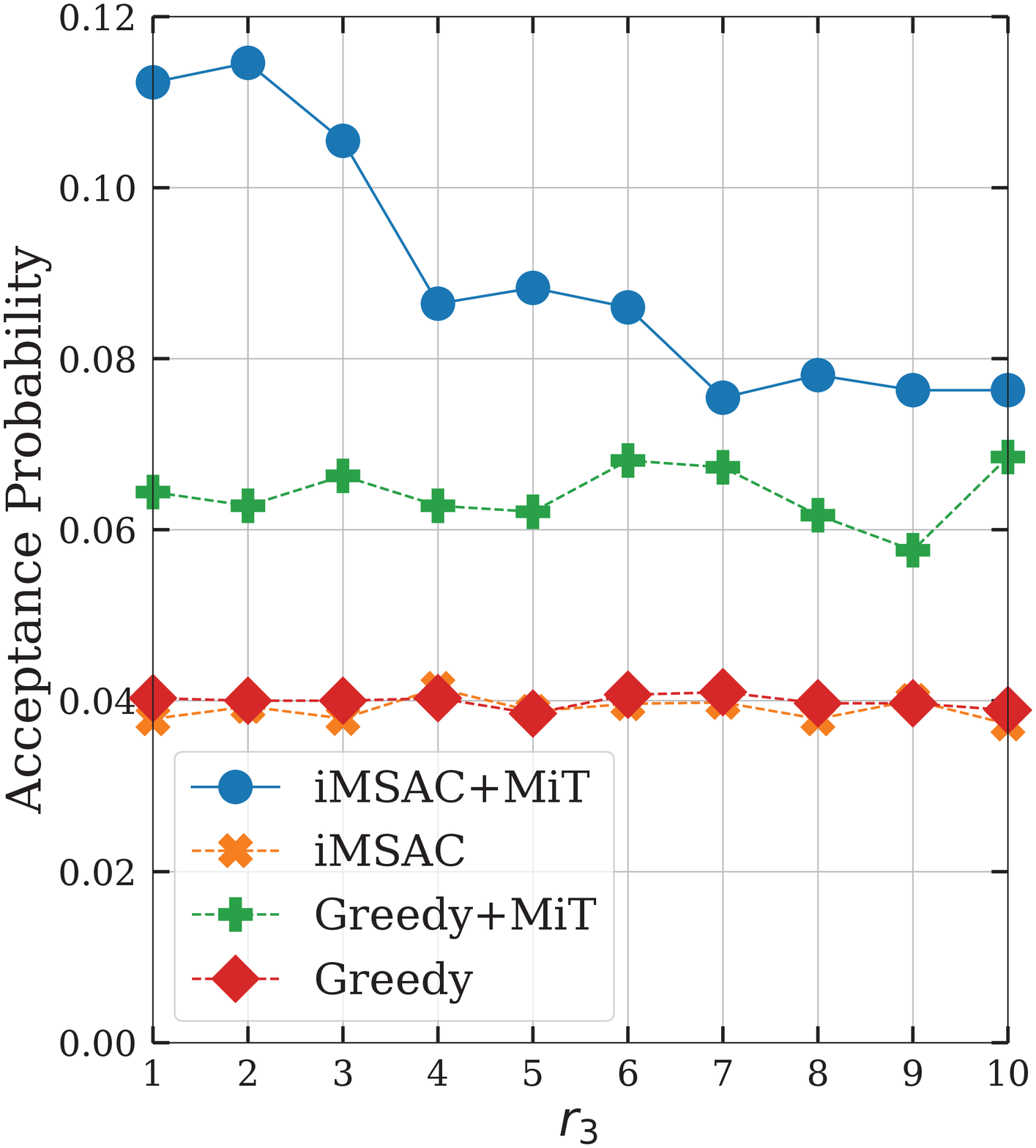}
			&\includegraphics[width=0.25\linewidth]{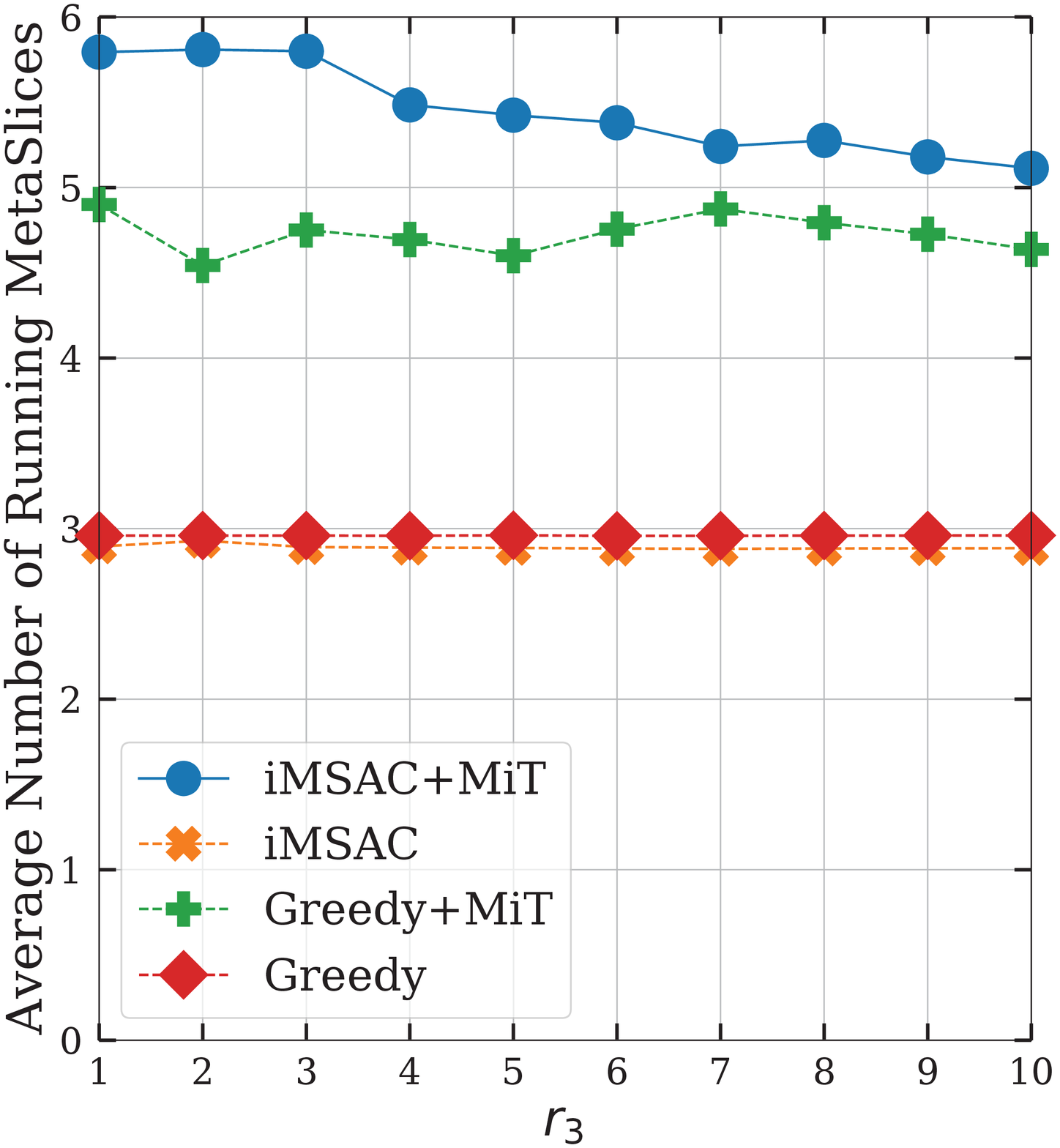}\\
			\text{(a) Average rewards}&\text{(b)  Acceptance probability }&\text{(c) Average number of running MetaSlices}
		\end{array}$
		\caption{Vary the immediate reward of class-3.}
		\label{fig:vary_fee1}
	\end{figure*}	
	We now investigate the robustness of our proposed solution, i.e.,  iMSAC+MiT, in different scenarios.
		First, we vary the storage, radio, and computing resources from $400$ GB, $400$ MHz, and $400$ GFLOPS/s to $2200$ GB, $2200$ MHz, and $2200$ GFLOPS/s, respectively.
		In other words, the total number of functions supported by the system is varied from $10$ to $55$.
		The policies of iMSAC+MiT and iMSAC are obtained after $3.75\!\times\!10^5$ learning iterations. 
		In this scenario, two metrics for evaluating the Admission Controller's performance are average reward and acceptance probability since they clearly show the effectiveness of the admission policy in terms of the income for MISP (i.e., average reward) and the service availability for end-users.
		Figure~\ref{fig:vary_vms}(a) clearly shows that as the total amount of system resources increases, the average rewards of all approaches increase.
		This is due to the fact that the higher the total resource is, the higher number of MetaSlice the system can host, and thus the greater revenue the system can achieve.  
		It is observed that the proposed algorithm iMSAC+MiT always obtains the highest average reward, up to $80\%$ greater than that of the second-best policy in this scenario, i.e., Greedy+MiT.
		Similarly, Fig.~\ref{fig:vary_vms}(b) shows that iMSAC+MiT achieves the highest acceptance probability for an arriving MetaSlice request, up to $47\%$ greater than that of the Greedy+MiT, i.e., the second-best policy.
		In addition, Figs.~\ref{fig:vary_vms}(a) and (b) demonstrate the benefit of MetaInstance.
		In particular, it helps the system to increase the average rewards and acceptance rates of both algorithms (i.e., iMSAC and Greedy) by up to $396\%$ and  $222\%$, respectively.

	 	\begin{figure*}[t]
	 		\centering
	 		$\begin{array}{cccc}		
	 			\includegraphics[width=0.233\linewidth]{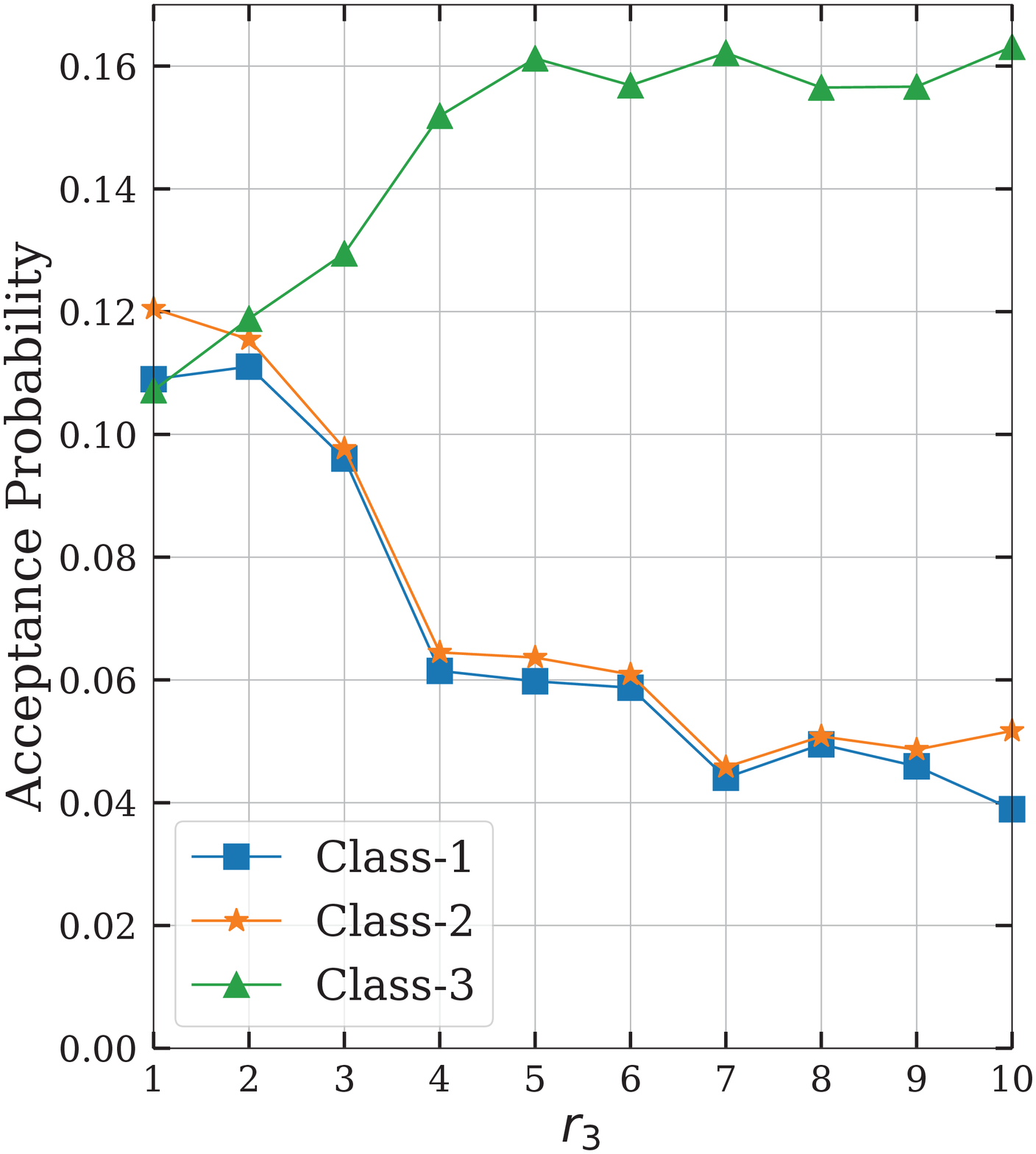}
	 			&\includegraphics[width=0.233\linewidth]{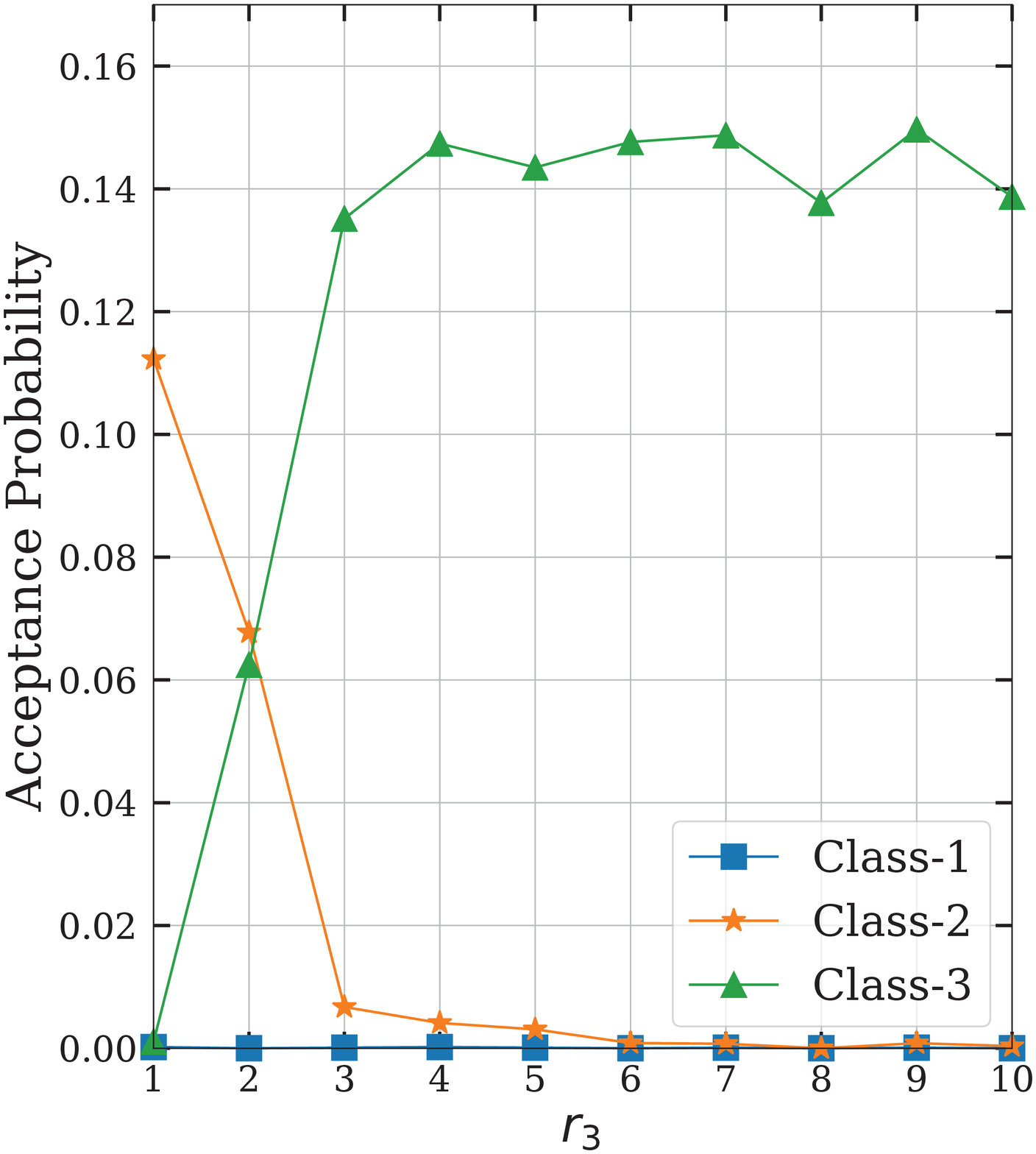} 
	 			&\includegraphics[width=0.233\linewidth]{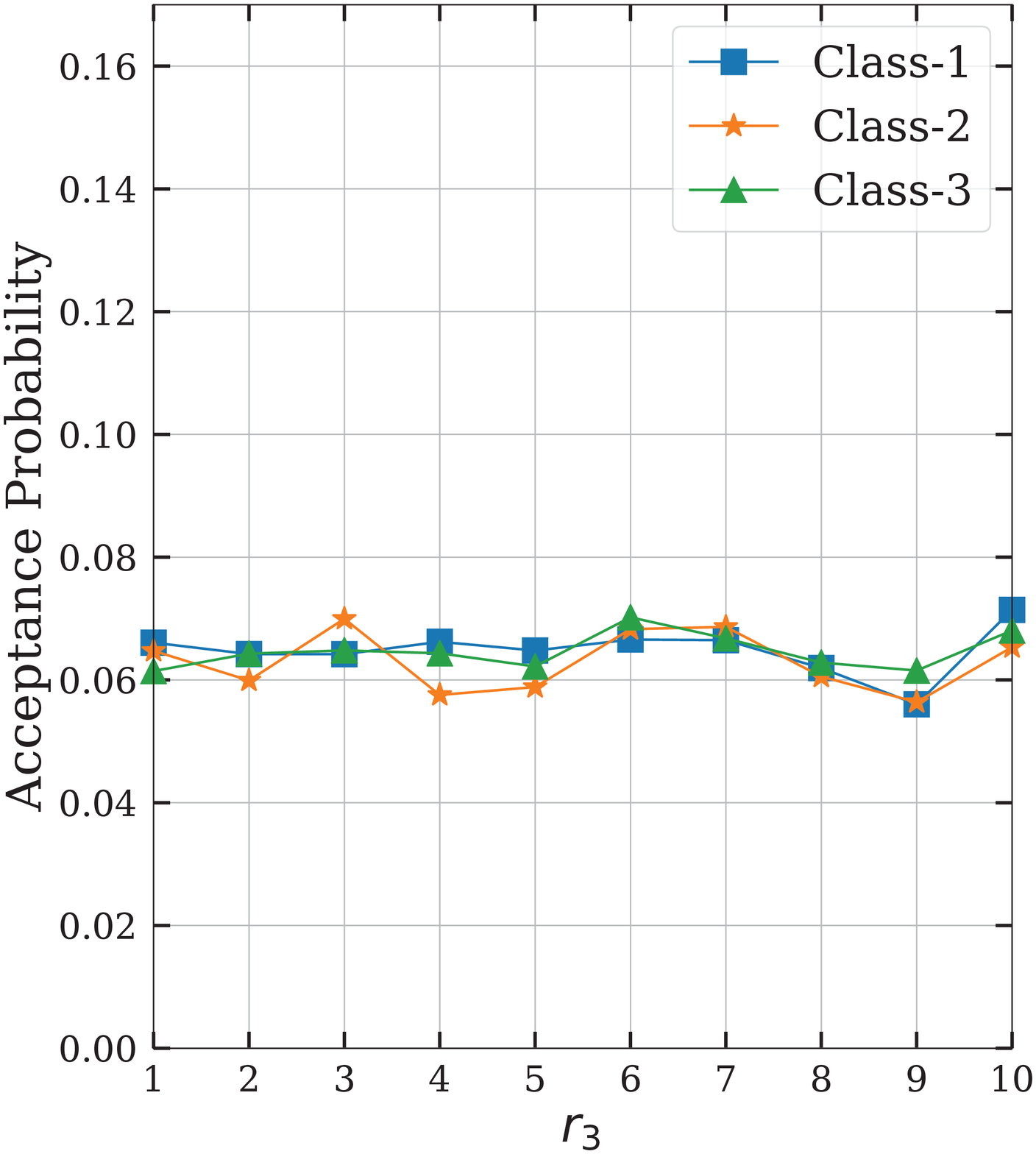} 
	 			&\includegraphics[width=0.233\linewidth]{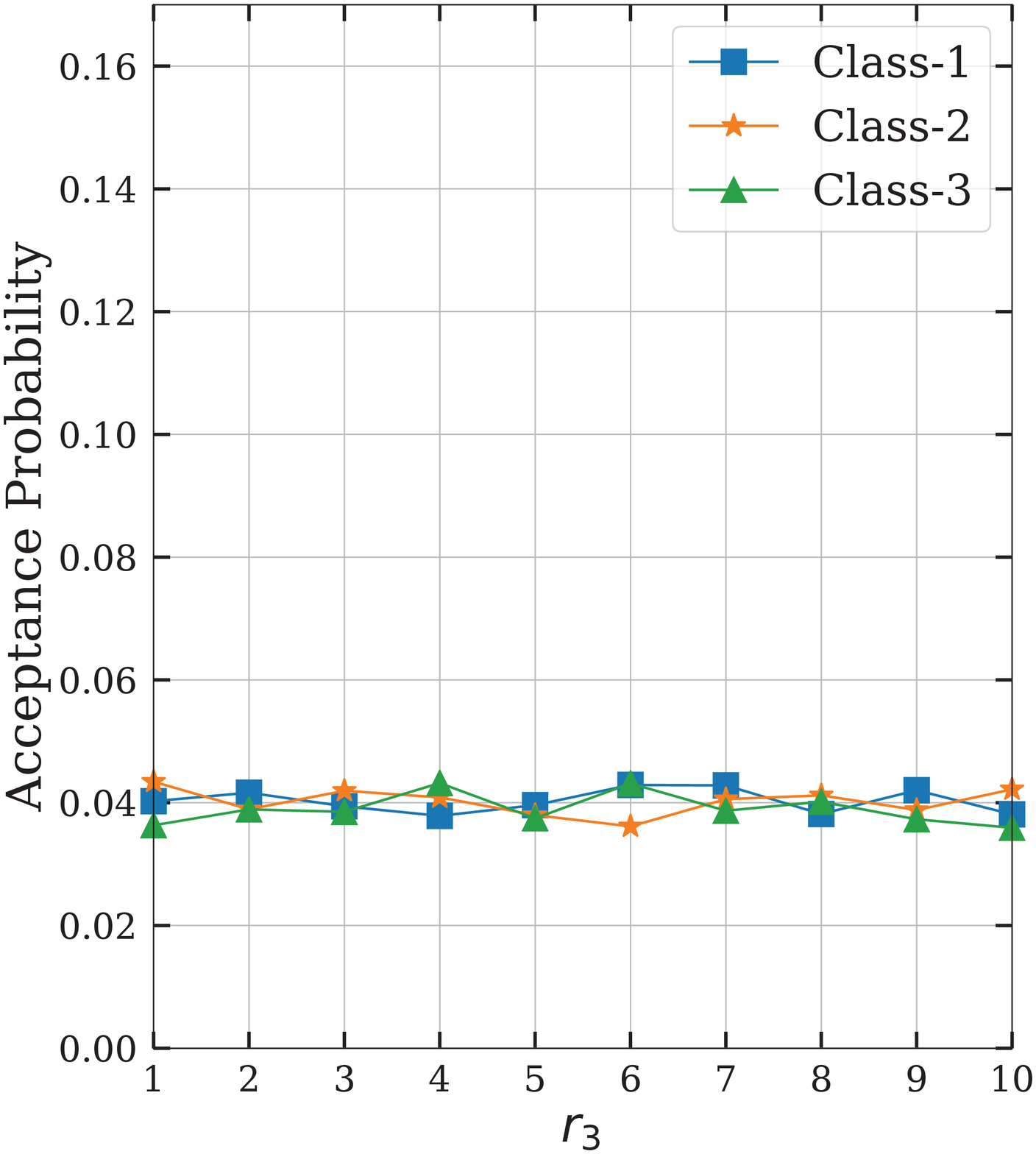}\\
	 			\text{(a) iMSAC+MiT}&\text{(b) iMSAC} &\text{(c) Greedy+MiT}&\text{(d) Greedy}
	 		\end{array}$
	 		\caption{The acceptance probability per class when varying the immediate reward of class-3.}
	 		\label{fig:vary_fee2}
	 	\end{figure*}
	 To gain more insights, we look further at the acceptance probability for each class of MetaSlice.
	 	As shown in Figs.~\ref{fig:vary_vms}(c) and (d), for the iMSAC+MiT and iMSAC, the acceptance probabilities of class-3 are always higher than that of other classes when the number of resources increases.
	 	Meanwhile, Greedy+MiT and Greedy accept requests from all the classes at almost the same probability, as depicted in Figs.~\ref{fig:vary_vms}(e), and (f). 	
	 	Recall that the arrival rate of class-3 is the lowest value (i.e., $\lambda_3\!=\!25$), while the immediate reward for accepting requests class-3 is the greatest value, i.e., $r_3\!=\!4$.	 	
	 	Thus, the proposed algorithm iMSAC can learn and adjust its policy to achieve the best result.
	 	More interestingly, for the iMSAC+MiT results, the acceptance probability of class-3 has significant gaps (up to $50\%$ greater than those of other classes) when the total resources are little (i.e., less than $20$), as shown in Fig.~\ref{fig:vary_vms}.
	 	This stems from the fact that if the available resources are low, the Admission Controller should reserve resources for future requests from class-3 with the highest reward.
	 	In contrast, if the system has more available resources, the Admission Controller should accept requests from all classes more frequently. 
	 	This observation is also shown in Fig.~\ref{fig:vary_vms}(d), where the MetaInstance is not employed. 		

	Next, we investigate one of the most important factors in the MetaSlicing framework, which is the maximum number of MetaSlices that share the same function, denoted by $N_L$.
		In this experiment, we set the resources the same as those in Fig.~\ref{fig:convergene}, and other settings are set the same as those in Section~\ref{subsec:parameters}.
		Figure~\ref{fig:vary_app_limit}(a) shows that as the value of $N_L$ increases from $1$ to $15$, the average rewards obtained by our proposed solution iMSAC+MiT and Greedy+MiT first increase and then stabilize at around $0.88$.
		Remarkably, when the value of $N_L$ is small (i.e., less than $8$), the iMSAC+MiT's average reward is always greater than that of Greedy+MiT, up to $122\%$.
		The reason is that as $N_L$ increases, meaning that more MetaSlices can share the same function, the number of MetaSlices can be deployed in the system increases.
		In other words, the MetaSlicing system capacity increases according to the increase of $N_L$.
		As such, the Admission Controller can accept more requests to obtain greater rewards when $N_L$ increases.
		In addition, the thresholds in average rewards for both approaches originate from the fact that the arrival and departure processes of class $i$ follow the Poison and exponential distributions with a fixed mean ~$\lambda_i$ and $1/\mu_i$.
	
	In terms of the acceptance probability for a MetaSlice request, similar observations can be made in Fig.~\ref{fig:vary_app_limit}(b). 
		Specifically, our proposed solution maintains higher request acceptance probabilities (up to $34\%$) than that of the Greedy+MiT when $N_L$ is less than $8$.
		As $N_L$ increases, the acceptance probabilities of both approaches increase, then they are stable at around $0.46$ when \mbox{$N_L>4$} for iMSAC+MiT and $N_L>7$ for Greedy+MiT.
		The reasons are similar as those in Fig.~\ref{fig:vary_app_limit}(a) 
		In particular, the higher the system capacity is, the higher the request acceptance probability is.
		Unlike the above metrics, the average numbers of MetaInstances decrease for both approaches as the value of $N_L$ increases from $1$ to $15$, as shown in Fig.~\ref{fig:vary_app_limit}(c).
		The reason is that an increase of $N_L$ can result in increasing the number of MetaSlices in a MetaInstance, thereby decreasing the number of MetaInstances in a system with a fixed capacity.
		The above observations in Fig.~\ref{fig:vary_vms} and Fig.~\ref{fig:vary_app_limit} show the superiority of our proposed approach compared with others, especially when the system resources are very limited. 
	
	We continue evaluating our proposed solution in the case where the immediate reward of class-3, i.e., $r_3$,  is varied from $1$ to $10$.
		In this experiment, we set the storage, radio, and computing resources to $400$ GB, $400$ MHz, and $400$ GFLOPS/s, respectively.
		The arrival rate vector of MetaSlice is set to $\boldsymbol{\lambda}\!=\![60, 50, 40]$ to explore the robustness of our proposed solution.
		In Fig.~\ref{fig:vary_fee1}(a), as $r_3$ increases, the average rewards obtained by all approaches increase.
		In particular, the results demonstrate that our proposed solution, i.e., iMSAC+MiT, consistently achieves the highest average reward, up to $111\%$ greater than that of the second-best, i.e., Greedy+MiT when \mbox{$r_3\!=\!1$}.
		Interestingly, when $r_3$ is small (i.e., less than $4$), the iMSAC's average rewards are lower than those of the Greedy+MiT.
		However, when $r_3$ becomes larger than or equal to $4$, the average rewards obtained by iMSAC are higher than those of the Greedy+MiT.
	
	Similarly, Figs.~\ref{fig:vary_fee1}(b) and (c) show that our proposed solution always obtains the highest values compared to those of other approaches in terms of the acceptance probability and average number of running MetaSlices when $r_3$ increases from $1$ to $10$.
		Interestingly, even with a decrease in the acceptance probability and average number of running MetaSlices, the average rewards obtained by the iMSAC+MiT increase as $r_3$ is varied from $1$ to $10$, as shown in Fig.~\ref{fig:vary_fee1}. 
		The reason is that when the immediate reward of class-3 is very high (e.g., $r_3\!=\!10$) compared to those of class-1 and class-2 (i.e., $1$ and $2$, respectively), the iMSAC+MiT reserves more resources for the future requests of class-3.  
		
	We now further investigate the above observations when varying the immediate reward of requests class-3 by looking deeper at the acceptance probability per class for each approach.
		Figures~\ref{fig:vary_fee2}(a)-(d) illustrate the acceptance rate per class according to the policies obtained by the proposed and counterpart approaches.		
		In Fig.~\ref{fig:vary_fee2}(c), the Greedy+MiT's acceptance probabilities for all classes are almost the same, at around $0.06$, when the immediate reward of class-3 increases from $1$ to $10$.
		A similar trend is observed for the Greedy but at a lower value, i.e., $0.04$, in Fig.~\ref{fig:vary_fee2}(d).
		In contrast to Greedy and Greedy+MiT, iMSAC+MiT's acceptance probability for class-3 increases while those of other classes decrease as $r_3$ increases from one to $10$, as shown in Fig.~\ref{fig:vary_fee2}(a).
		More interestingly, when the immediate reward of class-3 requests  is small (i.e., $r_3\!<\!2$), class-3 requests have the lowest acceptance probability compared to those of other classes.
		However, when the immediate reward of class-3 requests is larger than or equal to $2$, class-3 requests will achieve the highest acceptance probability compared with those of other classes. 		
		Moreover, when $r_3>4$, the acceptance probability for class-3 requests obtained by the iMSAC+MiT is stable at around $0.16$.
		
	Similar to the iMSAC+MiT, the iMSAC's acceptance probability for class-3 requests also increases until reaching a threshold with a lower value, i.e., around $0.14$, compared to that of the iMSAC+MiT.
		Furthermore, the acceptance probability of class-3 requests obtained by the iMSAC-base solutions is up to four-time greater than those of the Greedy-based solutions.	
		Thus, the iMSAC+MiT and iMSAC can obtain a good policy in which the acceptance probability of a class increases if its reward increases compared with the rewards of other classes, and vice versa. 		 			 
		Note that our proposed solution does not need complete information about the MetaSlice's arrival and departure processes in advance.
		However, as observed, the proposed solution can always achieve the best results in all scenarios when we vary important parameters of the system. 

	\section{Conclusion}
	In this paper, we have proposed two innovative techniques, i.e., the application decomposition and the MetaInstance, to maximize resource utilization for the Metaverse built on a multi-tier computing architecture.
		Based on these techniques, we have developed a novel framework for the Metaverse, i.e., MetaSlicing, that can smartly allocate resources for MetaSlices, i.e., Metaverse applications, to maximize the system performance.
		Moreover, we have proposed a highly effective framework based on sMDP together with an intelligent algorithm, i.e., iMSAC, to find the optimal admission policy for the Admission Controller under the high dynamics and uncertainty of resource demands.
		The extensive simulation results have clearly demonstrated the robustness and superiority of our proposed solution compared with the counterpart methods as well as revealed key factors determining the system performance.
				   	
	
	\ifCLASSOPTIONcaptionsoff
	\newpage
	\fi
	
	
	
	\bibliographystyle{IEEEtran}
	\bibliography{refs_mts}

\appendices
	\section{The Proof of Theorem 1}
	\label{ap:proof_theorem1}
	First, we need to prove the existence of the limiting matrix $\overline{\mathcal{T}}_\pi$ given in \eqref{eq:limiting_matrix}.
	In~\cite{puterman_markov_1994}, it is proven that for an aperiodic irreducible chain (as that of our proposed sMDP), the limiting matrix (which is the Cesaro limit of order zero named \textit{C-lim}) of the transition matrix, i.e.,$\mathcal{T}_\pi^g$,  exists as follows:
	\begin{equation}
		\label{eq:Cesaro_limit}
		\overline{\mathcal{T}}_\pi = \mbox{\textit{C-lim}} =\lim_{G \rightarrow \infty} \frac{1}{G} \sum_{g=0}^{G-1} \mathcal{T}_\pi^g.
	\end{equation}
	
	Next, because the total probabilities that a given state moves to others is one, i.e., \mbox{$\sum_{\mathbf{s}' \in \mathcal{S}}\mathcal{T}_\pi(\mathbf{s}|\mathbf{s}')=1$},  we have:	
	\begin{align}
		\overline{\mathcal{T}}_\pi r(\mathbf{s},\pi(\mathbf{s})) &=	\lim_{G \rightarrow \infty} \frac{1}{G+1} \mathbb{E} \left[\sum_{g=0}^{G} r(\mathbf{s}_g,\pi(\mathbf{s}_g)) \right], \forall \mathbf{s} \in \mathcal{S}, \\
		\overline{\mathcal{T}}_\pi y(\mathbf{s}, \pi(\mathbf{s})) &= \lim_{G \rightarrow \infty} \frac{1}{G+1} \mathbb{E} \left[\sum_{g=0}^{G} \tau_g \right] \forall \mathbf{s} \in \mathcal{S}.
	\end{align}	
	It is observed that the long-term average reward $\mathcal{R}_\pi(\mathbf{s})$ in~\eqref{eq:average_reward_function} is the ratio between $\overline{\mathcal{T}}_\pi r(\mathbf{s},\pi(\mathbf{s}))$ and $\overline{\mathcal{T}}_\pi y(\mathbf{s}, \pi(\mathbf{s}))$. 
		Note that by the quotient law for limits, the limit of a division of two functions is equivalent to the division of the limit of each function if the limit of a function at the denominator is not equal to zero.
		Since $\tau_g$, i.e., the interval time between two consecutive decision epochs, is always larger than zero, $\overline{\mathcal{T}}_\pi y(\mathbf{s}, \pi(\mathbf{s}))$ is always larger than zero. 
	Therefore, the long-term average reward $\mathcal{R}_\pi(\mathbf{s})$ exists. 
	
	\section{The Proof of Theorem 2}
	\label{ap:proof_theorem2}
	This proof is based on the irreducible property of the underlying Markov chain, which is proven as follows.
	As mentioned in Section~\ref{subsec:state_space}, the considered state space~$\mathcal{S}$ consists of the currently available resources in the system, the required resources, the class ID, and the similarity of the MetaSlice request.
	Recall that the request arrival and MetaSlice departure follow the Poisson and exponential distributions, respectively.
	In addition, the requested MetaSlice can have any function supported by the MetaSlicing system, and thus its required resources are arbitrary.
	Given the above, suppose that the Admission Controller observes state $\mathbf{s}$ at time $t$, then the system state can move to any other state $\mathbf{s}' \in \mathcal{S}$ after a finite time step.
	Therefore, the proposed sMDP is irreducible, and thus the long-term average reward function $\mathcal{R}_\pi(\mathbf{s})$ is well-defined regardless of the initial state and under any policy.
	
\end{document}